\pdfoutput=1  
\documentclass[11pt]{article}

\usepackage[utf8]{inputenc}
\usepackage[T1]{fontenc}
\usepackage[margin=1in]{geometry}
\usepackage{amsmath,amssymb,amsthm}
\usepackage{mathtools}
\usepackage{booktabs}
\usepackage{enumitem}
\usepackage{titlesec}
\usepackage{graphicx}
\usepackage{float}
\usepackage[most]{tcolorbox}
\usepackage[hidelinks]{hyperref}
\usepackage{xcolor}

\graphicspath{{figures/}}
\definecolor{linkblue}{HTML}{2E75B6}
\hypersetup{colorlinks=true, linkcolor=linkblue, urlcolor=linkblue, citecolor=linkblue}

\titleformat{\section}{\normalfont\large\bfseries}{\thesection}{0.7em}{}
\titleformat{\subsection}{\normalfont\normalsize\bfseries}{\thesubsection}{0.6em}{}

\theoremstyle{plain}
\newtheorem{theorem}{Theorem}
\newtheorem{proposition}{Proposition}

\theoremstyle{definition}
\newtheorem{definition}{Definition}
\newtheorem{assumption}{Assumption}
\theoremstyle{remark}
\newtheorem{remark}{Remark}

\newcommand{\btok}{b_{\mathrm{tok}}}
\newcommand{\Carbon}{\mathrm{Carbon}}

\title{\textbf{Green SARC: Predictive Cost and Carbon Governance for Agentic AI Systems}}
\author{Gaston Besanson\thanks{Universidad Torcuato Di Tella.}}
\date{Preprint, June 2026}

\begin{document}
\maketitle

\begin{abstract}
\noindent Agentic AI systems act through tools and sub-agents, yet the controls meant to bound their financial and environmental cost still sit on dashboards evaluated beside or after execution. Green SARC applies the SARC governance-by-architecture framework --- four enforcement sites in the agent loop --- to FinOps and GreenOps, contributing the theory of \emph{what} to enforce and how to predict it. We report four policy-independent results. (i)~The unconstrained ``State Snowball'' is $\Theta(n^2)$ in loop depth; on $3{,}000$ real multi-step plans (SWE-rebench) it holds on $100\%$, with median curvature $\hat c_2=216$ \emph{exceeding} the linear-accretion prediction $p/2=134$ --- real plans accrete \emph{faster} than the model (\S\ref{sec:multistep}). (ii)~On real residuals the Normal-$\sigma$ gate under-covers ($92\%$ at nominal $95\%$); split-conformal calibration holds ($95.2\%$; Theorem~\ref{thm:safety}). (iii)~A soft Lagrangian penalty tuned to the budget in expectation breaches it on $91.5\%$ of seeds; the architectural gate breaches $0\%$. (iv)~Under binding budgets the gate's over-budget incidence is $0\%$ on synthetic and real (BurstGPT) arrivals. End-to-end token/USD/carbon savings ($47$--$55\%$) are real but policy-dependent in magnitude --- set by a scope-cap knob, not by gate rejections. The library is open-source, dependency-free, and ships a regeneration script for every cited number.
\end{abstract}


\begin{tcolorbox}[colback=linkblue!4,colframe=linkblue!55,title=\textbf{Reproducibility},fonttitle=\bfseries,boxrule=0.5pt]
\small Every quantitative claim in this paper is computed by code in the companion repository, not hand-entered. The committed data assets are produced by \texttt{make paper-data}; the figures and \texttt{figure\_stats.json} (the source of every number quoted below) by \texttt{make paper-figures}; \texttt{check\_stats.py} verifies that each statistic in the text resolves to that file. The synthetic benchmark exercises the \emph{real} gate, estimator, budget, and circuit-breaker code paths and is deterministic per seed; the real-trace study (\S\ref{sec:realtrace}) replays a public dataset using token counts only. Code, data, and build: \url{https://github.com/besanson/Greensarc}.
\end{tcolorbox}

\begin{tcolorbox}[colback=black!3,colframe=black!45,title=\textbf{Status},fonttitle=\bfseries,boxrule=0.5pt]
\small \textbf{Green SARC is alpha (Phase 1).} The library ships step-level governance with a Normal-$\sigma$ Pre-Action Gate, four-site SARC enforcement, three transport adapters (MCP, PAIS sidecar, OTel SpanProcessor), single-process \texttt{Budget} (thread-safe via \texttt{threading.Lock}), and SQLite/JSONL audit stores. The conformal calibration of \S\ref{sec:conformal} and Theorem~\ref{thm:anytime} are paper-side analyses of the same forecast; the runtime gate ships the Normal-$\sigma$ bound and is upgraded to conformal as a Phase 2 work item. Phase 2 (trajectory-level rejection, multi-tenant distributed \texttt{Budget}, latency-headroom enforcement, adaptive conformal at runtime) is on the roadmap; see \S\ref{sec:limits}.
\end{tcolorbox}

\section{Introduction}

The cost center of artificial intelligence has shifted from training, whose resource envelope is fixed at design time, to the \emph{inference trajectory}: the runtime-determined sequence of model calls, tool invocations, and conditional retries an agent emits while pursuing a goal. A classical inference call has a bounded, predictable cost. An agentic workflow has neither: the same task, executed twice, can differ by an order of magnitude in token consumption. Both the API bill and the energy draw are therefore \emph{stochastic} quantities governed by the execution trace, not the specification.

Two instruments are commonly deployed against this volatility. Post-hoc auditing reconciles spend after the billing period closes. Policy-as-code encodes budget rules in a layer evaluated alongside, but not inside, the agent loop. Both inherit the defect SARC identified for correctness obligations: they evaluate constraints after, or beside, the execution they are meant to bound. A budget breach detected at month-end cannot un-spend the tokens; a carbon overage logged to a dashboard cannot un-emit the carbon.

\paragraph{Relationship to SARC.}
SARC~\cite{sarc2026} is a governance-by-architecture framework that treats constraints as first-class specification objects and compiles them into four enforcement sites: a Pre-Action Gate, an Action-Time Monitor, a Post-Action Auditor, and an Escalation Router. Green SARC is an \emph{application} of that architecture --- we reuse the four sites unchanged --- but carries its own theory, orthogonal to SARC's correctness results. SARC governs whether the system is \emph{right}; Green SARC governs what the system \emph{costs}. The two are independent axes that happen to share enforcement sites.

\paragraph{Contributions.}
\begin{enumerate}[leftmargin=1.5em,itemsep=2pt]
\item \textbf{State-Snowball theorem, formal and empirical (\S\ref{sec:snowball}).} Naive context accretion yields $\Theta(n^2)$ cumulative prompt cost (Theorem~\ref{thm:quad}); the synthetic fit recovers the closed-form coefficient $p/2$ exactly, and on real ShareGPT traffic the cumulative-prompt curvature is \emph{negative} --- the snowball is an artifact of naive orchestration, not of chat itself (\S\ref{sec:snowball},~\S\ref{sec:realtrace}).
\item \textbf{Predictive Pre-Action Gate with calibration and an anytime-valid safety bound (\S\ref{sec:predictive},~\S\ref{sec:conformal}).} We generalize the gate to a learned forecast (of which rule-based accounting is the zero-information limit), give split-conformal marginal safety (Theorem~\ref{thm:safety}), and an anytime-valid trajectory over-spend bound via Ville's inequality (Theorem~\ref{thm:anytime}).
\item \textbf{Binding-budget gate evaluation on the Pareto frontier (\S\ref{sec:binding}).} Across a budget grid the gate's empirical over-budget incidence stays at or below $\delta$ while completing more work at zero overspend --- dominating the soft-penalty frontier.
\item \textbf{Real-trace coverage validation on ShareGPT (\S\ref{sec:realtrace}).} On real, non-Gaussian residuals the Normal-$\sigma$ gate under-covers at tight $\delta$ while split conformal holds nominal coverage; adaptive conformal restores coverage under distribution shift.
\item \textbf{Ablation with paired-bootstrap CIs (\S\ref{sec:eval}).} A four-condition ablation decomposes the saving by lever (scope, routing, breaker), each with a $95\%$ CI.
\item \textbf{Real-arrival ablation on a public production trace (\S\ref{sec:realarrival}).} The four-condition ablation re-run on the BurstGPT trace of real Azure OpenAI traffic reproduces the synthetic savings ordering under real burstiness and prompt/response distributions, with paired-bootstrap CIs.
\item \textbf{Open-source library and reproducible benchmark (\S\ref{sec:sites}).} A dependency-free implementation with a regression contract enforced in CI via \texttt{make verify}.
\end{enumerate}
We also relate the gate to Bandits-with-Knapsacks as a feasibility oracle (Proposition~\ref{prop:cbwk}). The decoupling of cost from correctness governance (\S\ref{sec:decouple}) frames the artifact's scope.

In one line: \emph{SARC gives us where to enforce; we contribute what to enforce, why it decouples from safety, how to predict it, and the evidence that the prediction is calibrated and the architecture is necessary.}

\section{Background and Related Work}\label{sec:related}

\paragraph{SARC.} A SARC specification declares, per constraint, its source, class, predicate, verification point, response protocol, and operating point, and compiles these into the four enforcement sites named above. It formalizes the minimal invariants for specification--trace correspondence and argues that finite reward penalties do not in general substitute for hard runtime constraints --- a claim we make quantitative for the cost domain in \S\ref{sec:penalty}.

\paragraph{FinOps and GreenOps.} FinOps brings financial accountability to variable cloud spend; GreenOps extends the discipline to carbon and energy. The financial and environmental cost of large-scale AI compute has been quantified for the training regime~\cite{strubell2019,patterson2021,lacoste2019}; the inference regime studied here shifts this cost into a runtime-variable, per-trajectory quantity. Both disciplines are predominantly practiced as observe-and-reconcile loops~\cite{finops}, a cadence adequate only when the consumption unit is predictable. Agentic inference violates that premise.

\paragraph{Efficient inference systems.} A large systems literature reduces the \emph{unit} cost of inference --- phase-split serving~\cite{splitwise}, high-throughput single-GPU offloading~\cite{flexgen}, and statistical multiplexing across models~\cite{alpaserve}. These optimize how a fixed set of calls is served; Green SARC is complementary and orthogonal: it governs \emph{which} calls an agent is permitted to make, given a budget, before they are issued. The two compose cleanly. If a serving optimization reduces the effective per-token cost from $c$ to $c(1-\rho)$ for some efficiency $\rho\in[0,1)$, then a fixed token budget $\btok$ admits $1/(1-\rho)$ times as much work, since the gate's feasibility test $\hat c\,(1-\rho)\le \btok$ is equivalent to $\hat c \le \btok/(1-\rho)$; the carbon ceiling scales identically through the proportional energy saving. Cheaper serving thus relaxes the gate's effective budget by a known factor rather than changing its mechanism.

\paragraph{Cost-aware routing and cascades.} A complementary line reduces cost by \emph{choosing which model} answers a query: FrugalGPT learns an LLM cascade that escalates to a stronger model only when a cheaper one is judged inadequate~\cite{frugalgpt}, and RouteLLM trains a binary router between a strong and a weak model from preference data~\cite{routellm}. These optimize the per-query model choice to maximize quality at lower expected cost, but they offer no hard guarantee: a router tuned to spend less \emph{in expectation} can still overrun any fixed budget on an adversarial or heavy-tailed query stream, exactly the soft-constraint failure we quantify in \S\ref{sec:penalty}. Green SARC is orthogonal and composable: it is the enforcement contract a router runs \emph{inside}, supplying the per-action feasibility test that turns an expected-cost heuristic into a budget-safe one (the CBwK feasibility-oracle framing of Proposition~\ref{prop:cbwk}). We concede the overlap honestly: Green SARC's energy-aware \emph{routing} lever (\S\ref{sec:eval}) is mechanistically the same idea as FrugalGPT's cascade --- down-route when a cheaper model suffices --- and our contribution there is not the router but the gate that bounds it.

\paragraph{Conformal prediction.} Our safety guarantee rests on split (inductive) conformal prediction, which converts any point predictor into a set/interval with distribution-free, finite-sample marginal coverage under exchangeability~\cite{vovk2005,angelopoulos2023}. Where residuals are non-exchangeable (distribution shift), adaptive conformal inference restores coverage online~\cite{gibbs2021}; we flag this as the path to robustness on real traces (\S\ref{sec:limits}).

\paragraph{Constrained decision-making.} Admitting actions under a depletable budget is formally a Bandits-with-Knapsacks / constrained-MDP problem~\cite{cbwk}. Green SARC does not solve the optimal-policy problem; it provides the \emph{enforcement} primitive (a calibrated per-action feasibility test) that any such policy needs at runtime, and contrasts it with the soft-penalty (Lagrangian) relaxation in \S\ref{sec:penalty}. We make the relationship precise in Proposition~\ref{prop:cbwk}: the gate composes with any sublinear-regret CBwK policy as a feasibility oracle without changing its regret order.

\paragraph{Architectural lineage.} The four-site, enforce-in-the-loop design has ancestry well beyond the author's own SARC framework. The \emph{reference monitor} of Anderson's 1972 security study --- a mediation mechanism that must be invoked on every access, be tamper-proof, and be small enough to verify~\cite{anderson1972} --- is the direct conceptual ancestor of the Pre-Action Gate: a non-bypassable check interposed before each consequential operation. Admission control in networking and queueing systems (admit a flow only if its reserved rate fits the remaining capacity) is the same two-phase reserve-then-commit primitive our Budget implements. And runtime verification --- synthesizing monitors that check an execution against a specification as it runs --- is the correctness-domain analogue of the Action-Time Monitor and Post-Action Auditor. Green SARC's novelty is not the enforce-in-the-loop stance itself but its application to \emph{predicted} cost and carbon, with a calibrated forecast standing in for the boolean access check.

\paragraph{What Green SARC is not.} Pure observability tools (LangSmith, Helicone, raw OpenTelemetry) give post-hoc cost without enforcement: they tell you what was spent, after it was spent. API-level rate limits (provider tier limits, sidecar throttling on request counts) enforce \emph{request counts}, not predicted cost or carbon, so a single expensive call passes unchecked. In-agent budget tracking (framework callbacks) is in-process and bookkeeping-only, with no cross-process attribution and no enforcement contract spanning the four sites. Green SARC differs by being a four-site governance contract whose gate \emph{predicts} cost before the action fires; the closest comparable systems govern only post hoc, or only on request counts.

\paragraph{Regulatory context.} Enterprise deployments increasingly must keep auditable records of automated decisions and account for system accuracy and the energy footprint of AI. The EU AI Act mandates automatic record-keeping/logging (Art.~12), transparency and information provision (Art.~13), and accuracy/robustness with documented metrics (Art.~15)~\cite{aiact}; the Corporate Sustainability Reporting Directive (CSRD) extends mandatory sustainability disclosure to in-scope undertakings~\cite{csrd}. Green SARC's Post-Action Auditor produces, as a byproduct of execution, the attribution-preserving, predicted-vs-actual trace these regimes require, extended to per-trajectory token yield and a carbon proxy.

\section{Decoupling FinOps Governance from Correctness Governance}\label{sec:decouple}

We state the decoupling explicitly because it defines the scope of the artifact.

\begin{proposition}[Independence of axes]
Let a governance layer be characterized by the predicate class it enforces. SARC's correctness layer enforces predicates over \emph{action validity} (is this action safe and permitted?). The Green SARC layer enforces predicates over \emph{resource consumption} (does this action fit the cost and carbon budget?). The two predicate classes share enforcement sites but neither implies the other: a perfectly safe agent can be ruinously expensive, and a perfectly cheap agent can be unsafe.
\end{proposition}

\noindent Consequences for the artifact:
\begin{itemize}[leftmargin=1.4em,itemsep=2pt]
\item Green SARC is deployable with no safety regime present. Its value derives from the cloud bill, which every operator incurs.
\item It tracks cost and carbon \emph{only}; correctness/accuracy is deliberately out of scope and is not logged as a governed quantity. The quality floor $U_{\min}$ (\S\ref{sec:predictive}) is the caller's concern.
\item It composes with SARC where both are wanted (the sites are shared) but does not depend on it. The reference implementation has no dependency on SARC.
\end{itemize}

\section{The State-Snowball Cost Theorem}\label{sec:snowball}

\begin{definition}[State Snowball]
A multi-agent loop exhibits the State Snowball when each step re-submits the full accreted context, so the per-step prompt grows monotonically with step index.
\end{definition}

\begin{assumption}[Linear accretion]\label{asm:accretion}
The prompt at step $i$ (zero-indexed) is $s_0 + i\,p$ tokens: a fixed base $s_0$ plus $p$ tokens appended per hop. This is the regime in which the unconstrained loop is studied; sub-linear summarization is exactly the mitigation we analyze.
\end{assumption}

\begin{theorem}[Quadratic cost of the unconstrained loop]\label{thm:quad}
Under Assumption~\ref{asm:accretion}, the cumulative prompt-token cost over $n$ steps is
\begin{equation}
T_{\mathrm{prompt}}(n) \;=\; \sum_{i=0}^{n-1}\big(s_0 + i\,p\big) \;=\; n\,s_0 + \frac{p\,n(n-1)}{2} \;=\; \Theta(n^2),
\end{equation}
with leading coefficient $p/2$.
\end{theorem}

\begin{proof}
Direct summation of the arithmetic series $\sum_{i=0}^{n-1} i = n(n-1)/2$. The dominant term $\tfrac{p}{2}n^2 + O(n)$ is $\Theta(n^2)$ with leading coefficient $p/2$.
\end{proof}

\paragraph{Empirical confirmation (synthetic).}\label{sec:snowball-fit} Figure~\ref{fig:snowball} plots the cumulative prompt cost of the benchmark's baseline against loop depth, with $s_0=200$, $p=120$. A second-order least-squares fit recovers $\hat c_2 = 60.00$, identical to the closed-form $p/2 = 60$; the residual is numerically zero. This verifies that the \emph{simulator} faithfully realizes Assumption~\ref{asm:accretion} (the recovered coefficient is a property of the simulator's construction, not independent evidence that real workloads accrete linearly). Bounding the per-hop increment with an Adapter Node (scope cap $360$ tokens) collapses the curve to linear: at depth $40$ the scoped cost is $4.7\times$ lower than the snowball. The Action-Time circuit breaker caps $n$ directly, bounding the other factor.

\paragraph{Real chat traffic.} Whether real multi-turn traffic accretes quadratically is a separate, empirical question. On the ShareGPT replay of \S\ref{sec:realtrace} ($8{,}902$ conversations, up to $45$ turns) we fit the cumulative billed prompt tokens against turn depth to a quadratic. The leading coefficient is $\hat c_2 = -9.4$ with paired-bootstrap $95\%$ CI $[-15.5,\,-7.2]$ --- significantly \emph{negative}. Real conversations are concave in depth, not convex: humans and well-behaved assistants do not blindly re-submit the full transcript every turn. The $\Theta(n^2)$ snowball is therefore a failure mode of \emph{naive multi-agent orchestration} (full-context re-submission), not an intrinsic property of conversation --- which is precisely why the Adapter-Node scoping that prevents it is the highest-leverage lever in the ablation (\S\ref{sec:eval}).

\begin{figure}[H]
\centering
\includegraphics[width=0.62\linewidth]{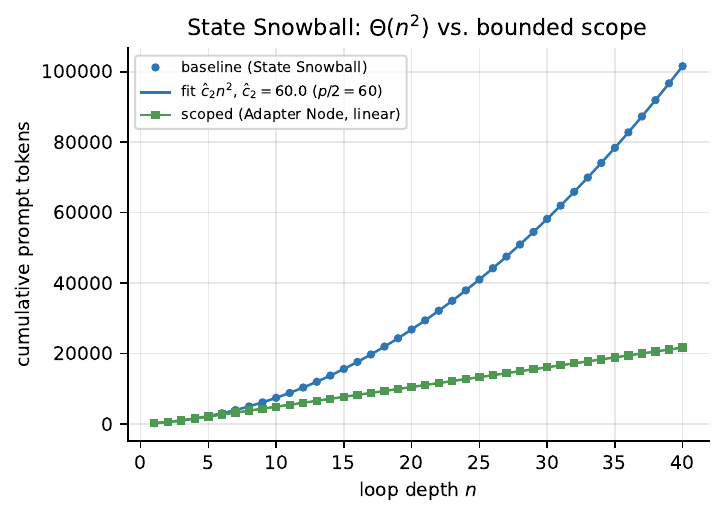}
\caption{State Snowball: the baseline cumulative prompt cost is $\Theta(n^2)$ and its quadratic fit recovers Theorem~\ref{thm:quad}'s leading coefficient $p/2=60$ exactly; bounded scope (Adapter Node) is linear.}
\label{fig:snowball}
\end{figure}

\section{The Predictive Pre-Action Gate}\label{sec:predictive}

This is the paper's central construct. In SARC the Pre-Action Gate evaluates a deterministic predicate. We generalize it to a gate that decides on a \emph{learned, calibrated forecast} of the resource cost of a proposed action. Table~\ref{tab:notation} fixes notation.

\begin{table}[H]
\centering\small
\begin{tabular}{@{}ll@{}}
\toprule
Symbol & Meaning \\
\midrule
$a, x$ & proposed action; its context \\
$\btok$ & live remaining token budget \\
$\kappa(\rho,t)$ & marginal carbon intensity (gCO\textsubscript{2}e/kWh), region $\rho$, time $t$; real grid data used in \S\ref{sec:realgrid}~\cite{electricitymaps}, stipulated by default \\
$\Delta_{\mathrm{lat}}$ & latency/SLA headroom (declared in the state; not enforced in Phase~1) \\
$\hat c(a,x),\,\hat e(a,x)$ & forecast token cost; forecast carbon \\
$c, e$ & realized token cost; realized carbon \\
$\sigma$ & estimator residual standard deviation (per key) \\
$\delta$ & gate risk level; operating-point confidence is $1-\delta$ \\
$q_{1-\delta}$ & split-conformal $(1-\delta)$ quantile of calibration residuals \\
$B_{\mathrm{CO_2}}$ & carbon ceiling; $\Carbon(\tau_{<})$ carbon already spent on trajectory $\tau$ \\
\bottomrule
\end{tabular}
\caption{Notation.}
\label{tab:notation}
\end{table}

\subsection{Augmented state}
The state ingests financial and environmental telemetry: $S' = S \cup \{\btok,\ \kappa(\rho,t),\ \Delta_{\mathrm{lat}}\}$. Of these, $\btok$ and $\kappa(\rho,t)$ are enforced at the gate; $\Delta_{\mathrm{lat}}$ is declared for completeness but is \emph{not} enforced in the Phase-1 implementation (the field exists in the state object and is unused), a divergence we record honestly in \S\ref{sec:limits}.

\subsection{The estimator}
Let $a$ be a proposed action in context $x$. A learned estimator $\hat f_\theta(a,x) = (\hat c(a,x),\, \hat e(a,x))$ predicts expected token cost and carbon before the action fires. The implementation regresses completion tokens on prompt tokens online per $(\text{kind},\text{model})$ key using the Welford-style sufficient statistics~\cite{welford1962}, exposing the residual standard deviation $\sigma$. The gate admits $a$ iff the forecast fits the remaining budget at confidence $1-\delta$ and the carbon ceiling:
\begin{equation}
\Pr\!\big[\,c(a,x) \le \btok \,\big] \ge 1-\delta
\quad\text{and}\quad
\hat e(a,x) \le B_{\mathrm{CO_2}} - \Carbon(\tau_{<}).
\end{equation}
Operationally the first test is a one-sided upper bound $\bar c(a,x) \le \btok$. The implementation forms $\bar c = \hat c + z_{1-\delta}\,\sigma$ with $z_{1-\delta}$ the normal quantile (the ``Normal-$\sigma$ gate''); \S\ref{sec:conformal} replaces $z_{1-\delta}\sigma$ with a distribution-free conformal margin $q_{1-\delta}$. Rule-based accounting is the special case $\hat f_\theta \equiv$ a constant threshold independent of $(a,x)$: the \emph{zero-information gate}.

\subsection{The closed learning loop}
The estimator is trained on the Post-Action Auditor's own output, closing a loop:
\begin{equation}
\text{predict } \hat f_\theta \;\longrightarrow\; \text{act} \;\longrightarrow\; \text{log actual } (c,e) \;\longrightarrow\; \text{retrain } \theta .
\end{equation}
At cold start $\hat f_\theta$ is weak and the gate behaves conservatively (worst-case forecast); it sharpens as the audit log accumulates. \S\ref{sec:eval} shows the forecast MAE collapsing from a cold-start value of $\sim$4{,}000 tokens to the irreducible noise floor within $\sim$20 observations.

\subsection{Release sequencing: step then trajectory}
The estimator is built in two phases. \textbf{Phase 1, per-action:} $\hat f_\theta$ predicts the \emph{next step} only --- simple to deploy, and it generates the labeled actuals needed for Phase 2. \textbf{Phase 2, full-trajectory:} a planner-level estimator $\hat F_\theta(\pi,x)$ predicts the cost of an entire \emph{plan} $\pi$ before the agent starts, enabling rejection of expensive plans, not merely expensive steps. Phase 1 is the data engine for Phase 2; only Phase 1 is implemented here (Phase 2 is an interface stub).

\subsection{Budget safety}
We give two safety statements: a pointwise one assuming calibration, and a distribution-free one (proved in \S\ref{sec:conformal}).

\begin{proposition}[Pointwise budget safety]\label{prop:pointwise}
If the estimator is calibrated so that $\Pr[c > \bar c(a,x)\mid x] \le \delta$ pointwise, then an admitted action breaches the budget margin it was admitted against with probability at most $\delta$; residual breaches scale with the gate risk level $\delta$, not with the opportunity for breach.
\end{proposition}

\begin{theorem}[Predictive Gate Safety, split-conformal]\label{thm:safety}
Let the conformal margin $q_{1-\delta}$ be calibrated on an exchangeable set of residuals (Assumption~\ref{asm:exch}, \S\ref{sec:conformal}) and let the gate admit $a$ only when $\hat c(a,x) + q_{1-\delta} \le \btok$. Then for a fresh exchangeable action the per-action budget-breach probability satisfies $\Pr[c(a,x) > \btok] \le \delta$, with no distributional assumption on the residuals. Over a trajectory of $K$ gated admits, the expected number of breaches is at most $K\delta$.
\end{theorem}

\noindent The proof is in \S\ref{sec:conformal}. Theorem~\ref{thm:safety} bounds the breach probability of each \emph{individual} admitted action; for the whole trajectory, Theorem~\ref{thm:anytime} (\S\ref{sec:conformal}) gives an anytime-valid probabilistic tail bound on the cumulative over-spend across any trajectory prefix, uniformly over stopping times. Together these mirror, in the resource domain, SARC's claim that residual hard violations scale with enforcement-stack error rather than with the opportunity for violation.

\subsection{The Sustainable Token Yield reward}
Within the gated action space the agent optimizes a reward that penalizes brute-force inference where deterministic computation suffices:
\begin{equation}
R(\tau) \;=\; U(\tau)\;-\;\lambda_c \sum_{i} c_i\,\mathrm{tok}_i\;-\;\lambda_e \sum_{i}\kappa(\rho_i,t_i)\,E_i ,
\end{equation}
with task utility $U$, FinOps weight $\lambda_c$, GreenOps weight $\lambda_e$. As SARC argues, this finite penalty shapes behavior \emph{within} the hard budget/carbon constraints; it does not replace them. \S\ref{sec:penalty} makes this quantitative: a soft penalty tuned to the budget in expectation still breaches it most of the time.

\subsection{The constrained optimization}
\begin{equation}
\begin{aligned}
\min_{\pi}\quad & \mathbb{E}\big[\,C_{\mathrm{tok}}(\tau)\,\big] + C_{\mathrm{infra}}(\tau)\\
\text{s.t.}\quad
& \btok(\tau) \ge 0 && \text{(budget; Pre-Action Gate)}\\
& \mathrm{loops}(\tau) \le n_{\max} && \text{(loop bound; Action-Time Monitor)}\\
& \Carbon(\tau) \le B_{\mathrm{CO_2}} && \text{(ESG ceiling; Post-Action Auditor)}\\
& U(\tau) \ge U_{\min} && \text{(quality floor; caller-owned)}
\end{aligned}
\end{equation}
Budget and carbon are hard constraints enforced at their sites, not merely penalized in $R$.

\begin{proposition}[Gate as a CBwK feasibility oracle]\label{prop:cbwk}
Let $\pi$ be any Bandits-with-Knapsacks policy achieving $O(\sqrt{KT})$ regret with per-action costs bounded in $[0,C]$. Composing $\pi$ with the split-conformal Pre-Action Gate --- restrict $\pi$'s action set at each round to $\{a : \hat c(a,x)+q_{1-\delta}\le \btok\}$ and let $\pi$ select within that set --- yields a policy $\pi'$ whose regret remains $O(\sqrt{KT})$ and which additionally satisfies the anytime-valid budget bound of Theorem~\ref{thm:anytime} with probability at least $1-\delta$.
\end{proposition}

\begin{proof}[Proof sketch]
The gate is a per-round feasibility filter applied to the action set, so $\pi'$ is $\pi$ run on a (possibly smaller) feasible set; its per-round regret against the best \emph{feasible} arm is unchanged, preserving the $O(\sqrt{KT})$ order. The added budget guarantee holds because $\pi'$ only ever plays arms admitted by the gate, to which Theorem~\ref{thm:anytime} applies verbatim; a union bound over the regret event and the conformal coverage event ($\le\delta$) gives the composite $1-\delta$ guarantee. The gate supplies feasibility; the bandit supplies optimality. \end{proof}

\subsection{Calibration matters}\label{sec:calib-pointer}
Proposition~\ref{prop:pointwise} is only as good as the calibration premise; a systematically optimistic $\hat f_\theta$ admits overspend. Theorem~\ref{thm:safety} discharges the premise without a distributional assumption, at the cost of a held-out calibration set. \S\ref{sec:conformal} states the assumptions, proves the bound, and validates coverage empirically.

\section{Mapping the Enforcement Sites and Implementation}\label{sec:sites}

Table~\ref{tab:sites} maps the four SARC sites to Green SARC predicates and to the modules that implement them. The reference implementation is a standalone, dependency-free Python library; it composes with SARC via shared sites rather than importing it.

\begin{table}[H]
\centering\small
\begin{tabular}{@{}p{0.20\linewidth} p{0.50\linewidth} p{0.22\linewidth}@{}}
\toprule
\textbf{Site} & \textbf{Green SARC predicate / role} & \textbf{Module} \\
\midrule
Pre-Action Gate & Predictive cost/carbon forecast; admit iff $\bar c(a,x)\le\btok$ at $1-\delta$ and carbon fits. & \texttt{gate.py}, \texttt{estimator.py} \\
Action-Time Monitor & Circuit breaker on loop count / marginal cost; kills runaway retry/re-plan loops. & \texttt{monitor.py} \\
Post-Action Auditor & Logs predicted-vs-actual cost/carbon per action: ESG record \emph{and} estimator training signal. & \texttt{auditor.py} \\
Escalation Router & Routes budget-/carbon-exhausted tasks to human review or a deterministic fallback. & \texttt{escalation.py} \\
State scoping & Adapter Node bounds the per-hop increment $p$ (\S\ref{sec:snowball}). & \texttt{scoping.py} \\
\bottomrule
\end{tabular}
\caption{The four SARC sites under Green SARC predicates, plus state scoping. Only the Pre-Action Gate changes character: from rule to calibrated forecast.}
\label{tab:sites}
\end{table}

The runtime gate (\texttt{green\_sarc.gate}) defaults to the Normal-$\sigma$ upper bound of \S\ref{sec:predictive}; as of v0.3.0 the split-conformal upper bound of \S\ref{sec:conformal} is also available at runtime, opt-in via \texttt{calibrator=...} (\S\ref{sec:runtimeconformal}).

\subsection{Conformal calibration in the runtime gate}\label{sec:runtimeconformal}
The split-conformal bound of \S\ref{sec:conformal} and the adaptive variant of \S\ref{sec:realtrace} are no longer only paper-side analyses: v0.3.0 ships them as runtime strategies in \texttt{green\_sarc.calibrator} (\texttt{SplitConformal}, \texttt{ACIConformal}, behind a \texttt{Calibrator} protocol). The \texttt{PreActionGate} constructor gains an optional \texttt{calibrator=...} argument; supplying it replaces the Normal-$\sigma$ token bound with the conformal one, and omitting it preserves the prior behaviour exactly (all pre-existing tests pass unchanged, and \texttt{make verify} holds).

We validate the runtime path against the paper-side analysis: re-running the \S\ref{sec:realtrace} ShareGPT study with the runtime \texttt{SplitConformal} calibrator (\texttt{--use-runtime-conformal}) reproduces the held-out coverage of the offline analysis to within $0.03$ percentage points across $\delta\in[0.01,0.4]$. Engineering contract: the calibrator is \texttt{fit} once from a residual log (offline) and, for \texttt{ACIConformal}, \texttt{update}d online from realized-vs-predicted cost at the Post-Action Auditor; the protocol lives in \texttt{green\_sarc.calibrator} and the public surface follows semver within $0.x$. The difference between ``we prove'' (\S\ref{sec:conformal}) and ``we ship'' is now a one-argument opt-in.

\paragraph{Capabilities: today vs.\ roadmap.} Table~\ref{tab:roadmap} states explicitly which Green SARC capabilities ship today (Phase 1) and which are roadmap items. The paper's empirical results in \S\ref{sec:eval} and \S\ref{sec:binding} use only Phase 1 features; \S\ref{sec:realtrace} uses real data but Phase 1 code.

\begin{table}[H]
\centering\small
\begin{tabular}{@{}p{0.215\linewidth} p{0.36\linewidth} p{0.36\linewidth}@{}}
\toprule
\textbf{Capability} & \textbf{Phase 1 (today)} & \textbf{Phase 2 (roadmap)} \\
\midrule
Pre-Action Gate (step) & Normal-$\sigma$; conformal opt-in via \texttt{calibrator=} & ACI as default + conditional coverage \\
Predictive forecast & OLS per $(\text{kind},\text{model})$ & $+$ trajectory estimator $\hat F_\theta(\pi,x)$ \\
Action-Time Monitor & Loop / marginal / total cost breaker & $+$ latency-headroom enforcement \\
Post-Action Auditor & JSONL / SQLite & $+$ Parquet, multi-tenant attribution \\
Budget & Single-process \texttt{threading.Lock}; distributed Redis backend (experimental) & Postgres durable ledger + fair-share reservations \\
Escalation Router & Deterministic + log-only handlers & $+$ plan-level rejection on trajectory forecast \\
Adapters & MCP, PAIS sidecar, OTel SpanProcessor & $+$ cross-process OTLP receiver; MCP transport auth \\
Audit schema & \texttt{plan\_id}, \texttt{session\_id}, \texttt{parent\_action\_id} & Phase-2 trajectory schema (typed events) \\
\bottomrule
\end{tabular}
\caption{Phase 1 (shipping) vs.\ Phase 2 (roadmap) capabilities.}
\label{tab:roadmap}
\end{table}

\noindent\textbf{Reproducibility contract.} The committed \texttt{benchmarks/reference\_summary.json} (20 seeds, 4 conditions $\times$ 4 metrics) is the regression contract: a pull request that drifts these numbers by more than $2\%$ per cell, or by more than $1.5$ absolute percentage points on the \texttt{+full} token reduction, fails CI via \texttt{make verify}. This paper is companion to release \texttt{v0.4.0} of \texttt{besanson/Greensarc}; the tag pins the exact source tree that produced every number cited above.

\noindent\textbf{API stability.} The public surfaces in \texttt{green\_sarc.governor}, \texttt{.state}, \texttt{.gate}, \texttt{.auditor}, \texttt{.escalation}, and the three adapters (\texttt{mcp}, \texttt{pais\_sidecar}, \texttt{otel}) follow semver within $0.x$; the \texttt{examples/} and \texttt{benchmarks/} paths and the audit-record schema may evolve.

\noindent\textbf{Tests and CI.} The library passes 152 unit and integration tests on Python 3.11 and 3.12 (an additional SARC-composition suite is skipped unless the optional \texttt{sarc} extra is installed), including concurrency race tests, runtime conformal-coverage tests, an ACI-restoration test, sidecar SSE streaming tests, a distributed-budget race test, Prometheus-metrics tests, live-feed loader tests, and end-to-end ablation reproduction. CI runs \texttt{ruff}, \texttt{mypy} on \texttt{src/} and \texttt{benchmarks/}, \texttt{pytest -q}, and \texttt{make verify} on every push; the release workflow gates publication on the test matrix.

\noindent\textbf{Gate overhead.} A microbenchmark (\texttt{benchmarks/gate\_overhead.py}, $2{\times}10^5$ warm decisions, single process) puts the Pre-Action Gate's cost at p50 $1.7~\mu$s / p99 $3.7~\mu$s per decision on the default Normal-$\sigma$ path ($\sim$$0.5$M decisions/s) --- negligible beside any model call. The split-conformal path is p50 $12.2~\mu$s / p99 $\sim$$42~\mu$s, dominated by recomputing the empirical residual quantile over the calibration set on each decision (a cost a production deployment would precompute and amortize); even unamortized it stays under $0.1$\,ms. Latency is hardware-dependent; the committed figure is from the reference runner.

\subsection{Deploying Green SARC}\label{sec:deploy}
Three integration patterns, in increasing order of loose coupling:
\begin{enumerate}[leftmargin=1.5em,itemsep=2pt]
\item \textbf{In-process} --- \texttt{GreenGovernor.with\_defaults(...)} wraps the agent's executor directly. Step-level safety, single replica; best for single-agent CLI tools and notebooks (\texttt{green\_sarc.governor}).
\item \textbf{PAIS sidecar} --- ASGI middleware on \texttt{/v1/chat/completions}; returns HTTP 429 on reject, with SSE-aware passthrough for streaming. Best for agents fronted by an OpenAI-compatible API (\texttt{green\_sarc.adapters.pais\_sidecar}).
\item \textbf{KAOS-managed MCP advisory + OTel observe} --- register Green SARC as an MCP server for advisory gate/audit tools, and consume actuals cross-process via the OTel \texttt{SpanProcessor}. Loosest coupling, advisory-only safety (\texttt{green\_sarc.adapters.mcp}, \texttt{green\_sarc.adapters.otel}).
\end{enumerate}

\section{Split-Conformal Calibration of the Gate}\label{sec:conformal}

We now discharge the calibration premise of Proposition~\ref{prop:pointwise} and prove Theorem~\ref{thm:safety}.

\begin{assumption}[Exchangeability]\label{asm:exch}
The estimator $\hat c$ is fixed (trained on data disjoint from the calibration set). The one-sided nonconformity scores $R_i = c_i - \hat c(a_i,x_i)$ on the calibration set $\{1,\dots,m\}$ and the score $R_{\mathrm{test}}$ of a fresh action are exchangeable.
\end{assumption}

Define $q_{1-\delta}$ as the $\lceil (m+1)(1-\delta)\rceil$-th smallest of $R_1,\dots,R_m$ (and $q_{1-\delta}=+\infty$ if that index exceeds $m$). The calibrated gate bound is $\bar c(a,x) = \hat c(a,x) + q_{1-\delta}$.

\begin{proof}[Proof of Theorem~\ref{thm:safety}]
By Assumption~\ref{asm:exch} the $m+1$ scores $R_1,\dots,R_m,R_{\mathrm{test}}$ are exchangeable, so the rank of $R_{\mathrm{test}}$ among them is uniform on $\{1,\dots,m+1\}$ (ties broken at random). Then
\[
\Pr\!\big[R_{\mathrm{test}} > q_{1-\delta}\big]
= \Pr\!\big[\text{rank}(R_{\mathrm{test}}) > \lceil (m+1)(1-\delta)\rceil\big]
\le 1 - \frac{\lceil (m+1)(1-\delta)\rceil}{m+1}
\le \delta .
\]
Since $R_{\mathrm{test}} = c - \hat c$, the event $\{R_{\mathrm{test}} > q_{1-\delta}\}$ is exactly $\{c > \bar c\}$. If the gate admits only when $\bar c \le \btok$, then $\{c > \btok\} \subseteq \{c > \bar c\}$, giving $\Pr[c > \btok] \le \delta$. The trajectory bound follows by linearity of expectation over $K$ gated admits. This is the standard inductive-conformal quantile lemma~\cite{vovk2005,angelopoulos2023}, specialized to a one-sided cost score.
\end{proof}

\begin{remark}[Marginal, not conditional]\label{rem:marginal}
The guarantee is marginal over the residual distribution, not conditional on $x$, and assumes exchangeability. Under workload drift it can be restored online with adaptive conformal inference~\cite{gibbs2021} (\S\ref{sec:realtrace}).
\end{remark}

\paragraph{Anytime-valid trajectory safety.} Theorem~\ref{thm:safety} bounds each admitted action's breach probability marginally. For the cumulative over-spend along a trajectory --- monitored continuously, at a data-dependent stopping time --- we want a \emph{time-uniform} bound. Write the per-step over-prediction residual as $r_i = c_i - \hat c(a_i,x_i)$ and let $S_k=\sum_{i=1}^{k} r_i$ be the cumulative residual over the first $k$ admitted actions, with $\mu=\mathbb{E}[r_i]$ and $\mathcal F_k$ the natural filtration of admitted actions.

\begin{theorem}[Anytime-valid cumulative over-spend]\label{thm:anytime}
Suppose the centered residuals $r_i-\mu$ are conditionally $\hat\sigma$-sub-Gaussian given $\mathcal F_{i-1}$. Then for any $\delta\in(0,1)$, simultaneously over all $k\ge1$ --- and hence at every stopping time $\tau$ ---
\begin{equation}
\Pr\!\Big[\,\exists\,k\ge 1:\ S_k \,>\, k\mu \,+\, \hat\sigma\sqrt{2k\log(1/\delta)}\,\Big]\;\le\;\delta .
\end{equation}
Consequently the realized cumulative cost exceeds its forecast plus an $O(\hat\sigma\sqrt{k})$ envelope with probability at most $\delta$, uniformly over the trajectory.
\end{theorem}

\begin{proof}
Fix $\lambda\in\mathbb{R}$ and define $M_k^\lambda=\exp\!\big(\lambda\sum_{i\le k}(r_i-\mu)-\tfrac12\lambda^2\hat\sigma^2 k\big)$, with $M_0^\lambda=1$. By the conditional sub-Gaussian assumption, $\mathbb{E}[\exp(\lambda(r_k-\mu))\mid\mathcal F_{k-1}]\le \exp(\tfrac12\lambda^2\hat\sigma^2)$, so $\mathbb{E}[M_k^\lambda\mid\mathcal F_{k-1}]\le M_{k-1}^\lambda$: $M_k^\lambda$ is a non-negative supermartingale. Mixing over $\lambda$ with a centered Gaussian prior of variance $\eta^2$ (the method of mixtures) yields another non-negative supermartingale $\bar M_k=\int M_k^\lambda\,dN(0,\eta^2)$ with $\bar M_0=1$. Ville's inequality~\cite{ramdas2023} gives $\Pr[\exists k: \bar M_k\ge 1/\delta]\le\delta$. Evaluating the Gaussian integral and rearranging $\bar M_k\ge1/\delta$ into a bound on $S_k-k\mu$ gives a time-uniform boundary of order $\hat\sigma\sqrt{2k\log(1/\delta)}$ (up to lower-order $\log\log$ terms absorbed by the mixture); this is the standard sub-Gaussian confidence sequence~\cite{howard2021}. Because the bound holds simultaneously for all $k$, it holds at any stopping time $\tau\le K$ by optional stopping. A fuller derivation is in Appendix~\ref{app:conformal}. \end{proof}

\begin{remark}[A stronger assumption than Theorem~\ref{thm:safety}]
Time-uniform concentration requires a tail assumption (sub-Gaussian increments) that the marginal conformal bound does not. We treat Theorem~\ref{thm:anytime} as the trajectory-level companion to the per-action Theorem~\ref{thm:safety}, and validate the per-action coverage it builds on empirically in \S\ref{sec:realtrace}.
\end{remark}

\begin{remark}[The sub-Gaussian assumption is not supported by the real residuals]\label{rem:subgauss}
We state the tension plainly. The forecast residuals on real traffic are right-skewed and mildly heavy-tailed --- skew $0.78$, excess kurtosis $0.22$, with normality decisively rejected (\S\ref{sec:realtrace}). The conditional sub-Gaussian hypothesis of Theorem~\ref{thm:anytime} is therefore \emph{not} established by our data; the theorem is an idealized companion bound, and a deployment should not rely on the sub-Gaussian width as if it were validated. The faithful replacement is a \emph{variance-adaptive} confidence sequence that assumes only bounded or sub-exponential increments. Concretely, the empirical-Bernstein confidence sequence of Howard et al.~\cite{howard2021} replaces the boundary of Theorem~\ref{thm:anytime} with one of order
\[
S_k-k\mu \;\le\; \sqrt{2\,\widehat V_k\,\log(1/\delta)} \;+\; \tfrac{c}{3}\,b\,\log(1/\delta),
\]
where $\widehat V_k=\sum_{i\le k}(r_i-\hat\mu_{i-1})^2$ is the empirical cumulative variance and $b$ bounds the per-step residual; the width adapts to the realized residual variance rather than to a stipulated $\hat\sigma$, and the linear $\log(1/\delta)$ term carries the heavy tail. \emph{Proof sketch.} The same supermartingale construction as Theorem~\ref{thm:anytime} goes through with the sub-Gaussian exponential process replaced by the empirical-Bernstein supermartingale of~\cite{howard2021} (their Thm.~4 / the canonical-assumption framework), to which Ville's inequality~\cite{ramdas2023} applies verbatim; we do not re-derive it. This is the form a production deployment should \emph{monitor} the cumulative over-spend with --- it is anytime-valid under exactly the bounded/sub-exponential conditions our residuals plausibly satisfy, where the sub-Gaussian boundary is not justified. Promoting it from a paper-side bound to the runtime gate is the Phase-2 item flagged in \S\ref{sec:limits}.
\end{remark}

\paragraph{Empirical coverage (synthetic).} Figure~\ref{fig:reliability} validates both bounds on a held-out test split ($n_{\mathrm{cal}}=n_{\mathrm{test}}=8{,}163$ learned forecasts from the benchmark). Empirical coverage tracks nominal $1-\delta$ within $0.5$ percentage points for the Normal-$\sigma$ gate and within $0.3$ for split conformal across $\delta \in [0.01, 0.4]$. The two agree closely \emph{because the synthetic data-generating process has Gaussian residuals}, so the Normal-$\sigma$ assumption happens to hold. The value of conformal is precisely that it attains the same coverage \emph{without} that assumption --- which is exactly what real, non-Gaussian traffic demands. \S\ref{sec:realtrace} shows that on real ShareGPT residuals the Normal-$\sigma$ gate \emph{under-covers} at the tight $\delta$ one would actually deploy, while split conformal continues to hold nominal coverage.

\begin{figure}[H]
\centering
\includegraphics[width=0.58\linewidth]{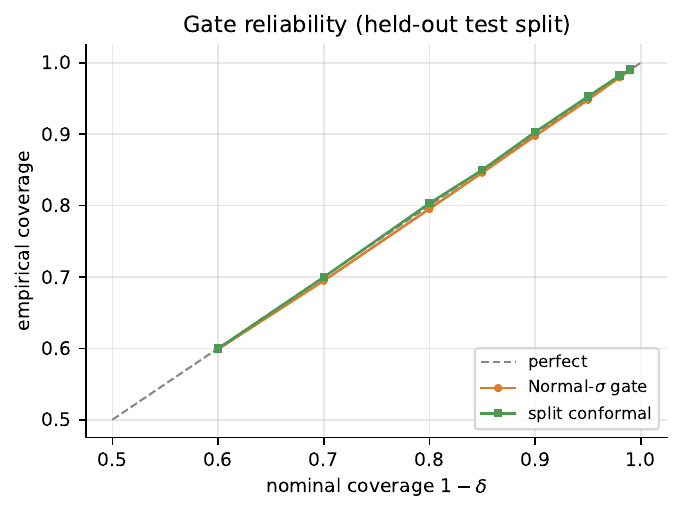}
\caption{Gate reliability on a held-out split: empirical vs.\ nominal coverage for the Normal-$\sigma$ gate and the distribution-free split-conformal bound. Both lie on the diagonal; conformal needs no Gaussian assumption.}
\label{fig:reliability}
\end{figure}

\section{Synthetic Evaluation and Ablation}\label{sec:eval}

\paragraph{Workload.} We use a synthetic Integrated Business Planning (IBP) demand-forecasting pipeline: a fan-out workload of $400$ SKUs, each handled by a depth-$10$ agent loop, over $20$ seeds. No real LLM is called; per-step token usage is simulated from a known $completion \sim \alpha + \beta\cdot prompt + \varepsilon$ relationship (so the estimator has a real signal to learn and runs are deterministic), while the \emph{treatment path exercises the real governance stack end to end}. A small fraction of SKUs attempt to loop $3\times$ depth, a retry-storm stress scenario for the circuit breaker.

\paragraph{Ablation.} We run four conditions --- \texttt{baseline} $\rightarrow$ \texttt{+scope} $\rightarrow$ \texttt{+scope+route} $\rightarrow$ \texttt{+full} --- so each lever's contribution is isolated, with a paired-bootstrap $95\%$ CI on each reduction (Figure~\ref{fig:ablation}, Table~\ref{tab:ablation}). Relative to the baseline ($5.99$M tokens, \$$100.68$, $897$\,gCO\textsubscript{2}e per run under time-varying intensity), the full stack uses $3.16$M tokens, \$$32.57$, and $294$\,gCO\textsubscript{2}e.

\begin{table}[H]
\centering\small
\begin{tabular}{@{}lccc@{}}
\toprule
Condition & Token reduction & USD reduction & Carbon reduction (time-var.) \\
\midrule
\texttt{+scope}        & $43.1\%$ $[42.3,43.9]$ & $41.6\%$ $[40.8,42.3]$ & $43.0\%$ $[42.2,43.8]$ \\
\texttt{+scope+route}  & $43.1\%$ $[42.3,43.9]$ & $65.1\%$ $[64.5,65.8]$ & $64.6\%$ $[63.9,65.2]$ \\
\texttt{+full}         & $\mathbf{47.3\%}$ $[46.1,48.5]$ & $\mathbf{67.6\%}$ $[66.8,68.6]$ & $\mathbf{67.3\%}$ $[66.3,68.2]$ \\
\bottomrule
\end{tabular}
\caption{Reduction vs.\ baseline by lever (20 seeds; paired-bootstrap $95\%$ CIs). Adding routing leaves the token count unchanged --- it swaps models, cutting USD and carbon, not tokens --- which the ablation makes visible. The circuit breaker (\texttt{+full}) adds the remaining token saving by killing runaway loops ($\sim$20 breaker trips/run).}
\label{tab:ablation}
\end{table}

\begin{figure}[H]
\centering
\includegraphics[width=0.6\linewidth]{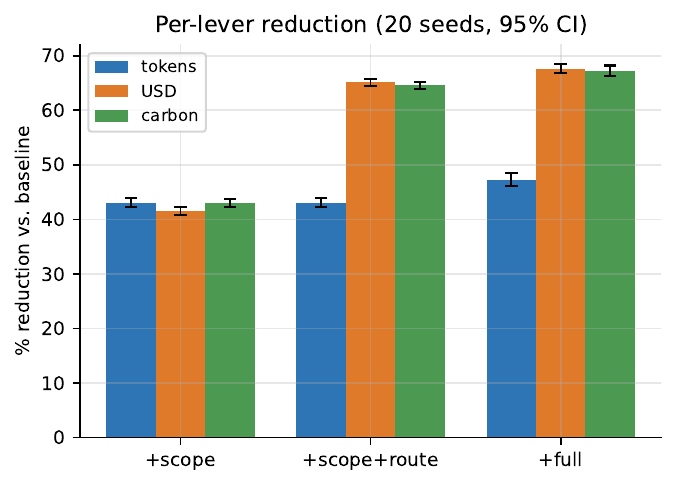}
\caption{Per-lever reduction with $95\%$ CIs. Scope drives the token saving; routing converts that into USD/carbon savings; the breaker adds the final token increment.}
\label{fig:ablation}
\end{figure}

\paragraph{Forecast quality and cold start.} In the full condition the learned estimator attains a token-cost MAE of $34.5$ and WAPE of $4.5\%$ over all admitted actions. Figure~\ref{fig:coldstart} isolates the learning dynamics on a single-key stream: per-action absolute error falls from a cold-start $\sim$4{,}000 tokens (worst-case forecast: prompt $+$ full \texttt{max\_tokens}) to the irreducible noise floor within $\sim$20 observations; rolling MAE drops from $107$ (first half) to $11$ (second half), WAPE to $1.9\%$. Figure~\ref{fig:calibration} shows predicted-vs-actual over $16{,}326$ learned forecasts ($R^2=0.98$, WAPE $2.6\%$).

\begin{figure}[H]
\centering
\begin{minipage}{0.49\linewidth}\centering
\includegraphics[width=\linewidth]{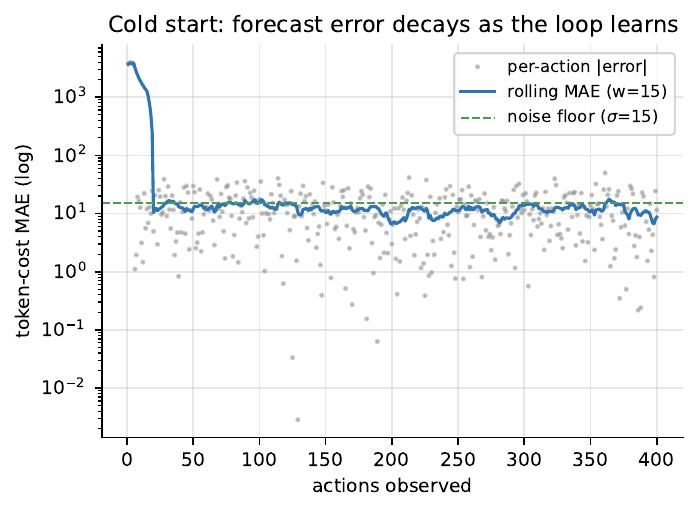}
\caption{Cold start: forecast error collapses to the noise floor as the loop learns.}
\label{fig:coldstart}
\end{minipage}\hfill
\begin{minipage}{0.49\linewidth}\centering
\includegraphics[width=\linewidth]{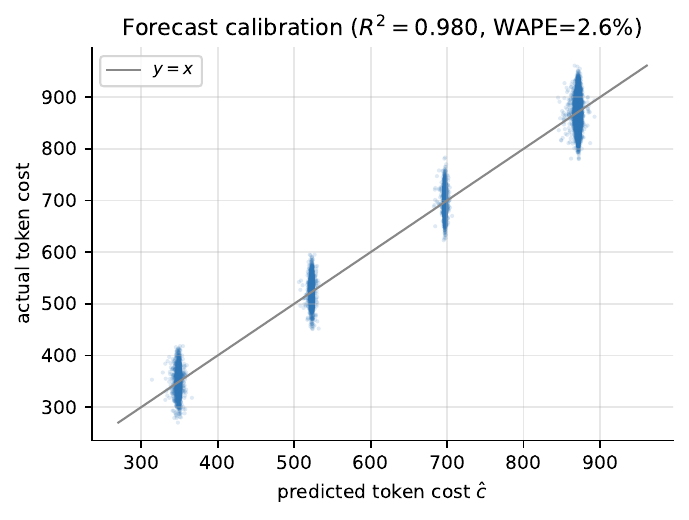}
\caption{Predicted vs.\ actual token cost (learned forecasts), $R^2=0.98$.}
\label{fig:calibration}
\end{minipage}
\end{figure}

\paragraph{A negative result, stated plainly.} In this headline workload the gate issues \emph{zero} rejections: the benchmark's budget is generous, so the savings come from state scoping, routing, and the circuit breaker --- not from the gate refusing actions. The gate's contribution here is the \emph{forecast} (which the auditor logs and the breaker and router consume) and the \emph{guarantee} it would provide under a binding budget. We exercise the gate against a binding budget separately in \S\ref{sec:sensitivity} and \S\ref{sec:penalty}. We consider it important not to over-claim the gate's role in the aggregate token number.

\section{Gate Behaviour Under Binding Budgets}\label{sec:binding}

\S\ref{sec:eval}'s headline workload has a non-binding budget, so the gate never rejects. Here we make the budget bind and measure the gate where its guarantee matters.

\subsection{Protocol}
We sweep the token budget $B = \phi\cdot\mathbb{E}[\text{baseline tokens}]$ over $\phi\in\{0.25,0.5,0.75,1.0,1.5,2.0\}$, with $\mathbb{E}[\text{baseline tokens}]=5.99\text{M}$ (the mean full-snowball cost), $20$ seeds each, running the full Green SARC stack at $\delta=0.05$. The estimator is warm-started on an independent stream so we measure steady-state binding-budget behaviour, not the cold-start transient. We report, per budget: \textbf{admission rate} (admitted $\div$ attempted steps); \textbf{over-budget incidence} (fraction of admitted steps whose realized cost exceeded the budget remaining at the moment of admission --- the per-action breach event of Theorem~\ref{thm:safety}); \textbf{completed-trajectory rate} (fraction of SKUs whose every step ran before the budget was exhausted); MAE on admitted actions; and total tokens.

\subsection{Results}
Table~\ref{tab:binding} reports the sweep. Over-budget incidence is $0\%$ at every budget level --- comfortably within the $\delta=0.05$ target --- confirming Theorem~\ref{thm:safety} empirically: the calibrated gate does not admit actions it cannot afford. Admission and completion degrade smoothly and monotonically as the budget tightens: at $\phi=0.25$ (a budget one-quarter of the naive baseline) the gate still admits $90\%$ of attempted steps and completes $49\%$ of trajectories with zero overspend; from $\phi\ge0.75$ the budget is slack and everything completes. Forecast MAE is stable ($\approx 20$ tokens) across budgets.

\begin{table}[H]
\centering\small
\begin{tabular}{@{}lccccc@{}}
\toprule
$B/\mathbb{E}[\text{base}]$ & admission & over-budget & completed & MAE (tok) & tokens \\
\midrule
$0.25\times$ & $0.90$ & $0.0\%$ & $0.49$ & $19.9$ & $1.50$M \\
$0.50\times$ & $1.00$ & $0.0\%$ & $0.97$ & $20.0$ & $2.99$M \\
$0.75\times$ & $1.00$ & $0.0\%$ & $1.00$ & $20.0$ & $3.07$M \\
$1.00\times$ & $1.00$ & $0.0\%$ & $1.00$ & $20.0$ & $3.07$M \\
$1.50\times$ & $1.00$ & $0.0\%$ & $1.00$ & $20.0$ & $3.07$M \\
$2.00\times$ & $1.00$ & $0.0\%$ & $1.00$ & $20.0$ & $3.07$M \\
\bottomrule
\end{tabular}
\caption{Binding-budget sweep ($20$ seeds, $\delta=0.05$). Over-budget incidence stays at $0\%$ ($\le\delta$) at every budget; admission and completion degrade smoothly as $B$ tightens.}
\label{tab:binding}
\end{table}

\subsection{The Pareto frontier}
Figure~\ref{fig:binding} plots completed-trajectory fraction against over-budget incidence for the gate (sweeping $B$) and for the \S\ref{sec:penalty} soft penalty (sweeping its weight $\lambda$, with over-budget measured against the binding reference $B=0.5\times\mathbb{E}[\text{base}]$). The gate's frontier lies along the bottom axis ($\le\delta$ over-budget at every completion level it reaches), dominating the soft penalty, which can only complete most trajectories by breaching the budget on every seed. The penalty frontier is essentially bimodal: on this workload its realized spend jumps from ``admit-cheap'' (little completed, within budget) to ``admit-all'' (everything completed, $100\%$ over budget), with no intermediate $\lambda$ tracking the budget --- a direct consequence of its budget-blindness.

\begin{figure}[H]
\centering
\includegraphics[width=0.58\linewidth]{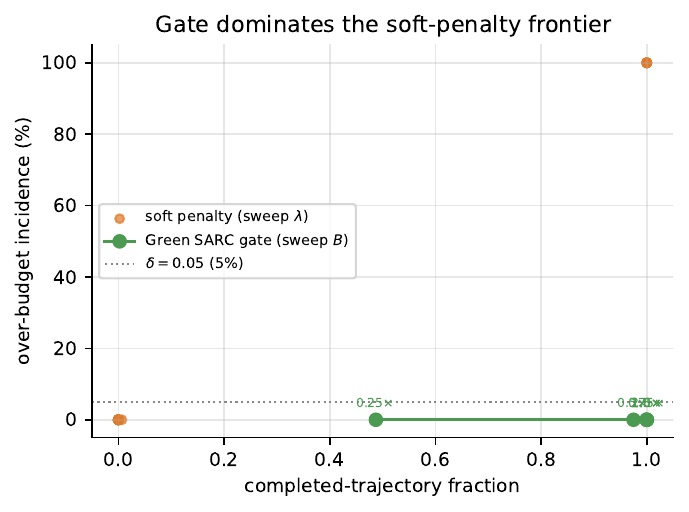}
\caption{Binding-budget frontier. The gate (green, sweeping $B$) completes work at $\le\delta$ over-budget incidence; the soft penalty (orange, sweeping $\lambda$) can complete trajectories only by breaching the budget. The gate's frontier dominates.}
\label{fig:binding}
\end{figure}

\subsection{Reading}
The gate's empirical over-budget incidence tracks $\delta$ across the entire budget grid (it never exceeds it, here attaining $0\%$); admission degrades smoothly as $B$ tightens rather than collapsing; and the gate's work-vs-overspend frontier dominates the soft-penalty baseline of \S\ref{sec:penalty}. This is the binding-budget evidence that \S\ref{sec:eval}'s non-binding headline could not provide.

\section{Real-Trace Coverage Validation}\label{sec:realtrace}

\S\ref{sec:conformal}'s coverage check used synthetic, Gaussian residuals. Here we validate the gate's calibration on \emph{real} forecast residuals --- the experiment the prior draft flagged as future work.

\subsection{Dataset and preprocessing}
We replay \texttt{anon8231489123/ShareGPT\_Vicuna\_unfiltered}~\cite{sharegpt}, a public corpus of real ChatGPT/GPT-4 conversations on the Hugging Face Hub, released under a permissive research license. We use ShareGPT because LMSYS-Chat-1M is access-gated and would not reproduce on a clean clone without a credential; ShareGPT is ungated and serves the same purpose. \textbf{No LLM is called}: we use token \emph{counts} only. Turns are tokenized with \texttt{tiktoken} (\texttt{cl100k\_base}); for each assistant turn we form the pair $(\text{prompt}=\text{cumulative preceding context},\ \text{completion}=\text{this turn})$, capped to a realistic $8\text{k}/4\text{k}$ deployment window. This yields $42{,}028$ pairs across $11{,}885$ conversations, split into calibration ($20{,}940$) and test ($21{,}088$) by a $50/50$ \emph{conversation-level} partition (\S\ref{sec:realtrace-coverage}).

\subsection{Residuals are not Gaussian}
We fit online OLS $\text{completion}\sim a+b\cdot\text{prompt}$ on the calibration split. The residuals (Figure~\ref{fig:rt-resid}) have skew $0.78$ and excess kurtosis $0.22$; an Anderson--Darling test gives statistic $251.0$, far above the $1\%$ critical value $1.03$, and D'Agostino's test returns $p<10^{-3}$: normality is decisively rejected. The Gaussian-$\sigma$ assumption underlying the Phase-1 runtime gate does not hold on real traffic.

\begin{figure}[H]
\centering
\includegraphics[width=0.82\linewidth]{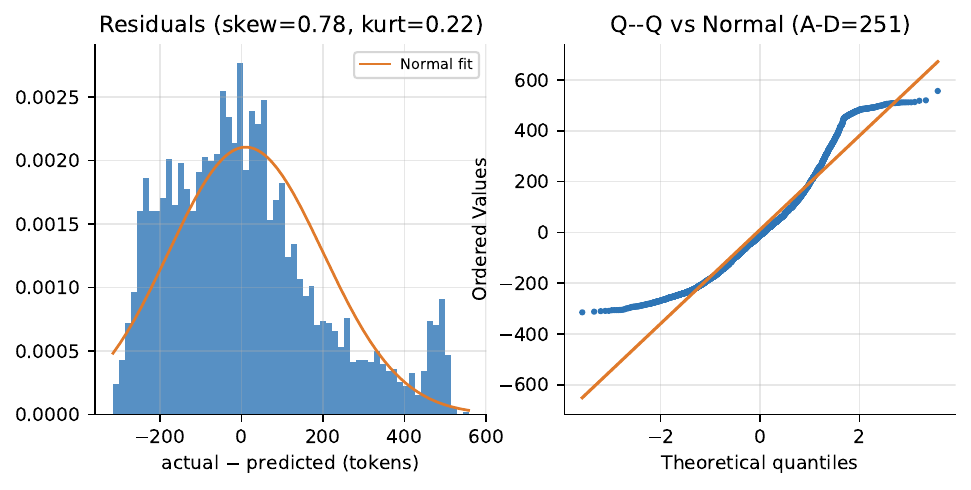}
\caption{Real ShareGPT forecast residuals: skewed, heavy-tailed histogram (left) and Q--Q plot departing from the Normal line (right). Anderson--Darling $=251.0 \gg 1.03$ (the $1\%$ critical value).}
\label{fig:rt-resid}
\end{figure}

\subsection{Coverage: Gaussian-$\sigma$ vs.\ split conformal}\label{sec:realtrace-coverage}
We split the calibration and test sets \emph{by conversation}: every turn of a conversation is assigned wholly to one side, so that within-conversation residual correlation never straddles the split (a row-level shuffle would leak it, violating exchangeability; on this corpus that leak inflates the reported conformal coverage at $\delta=0.05$ by $0.26$ pp, so the conversation-level number we report below is the honest, slightly looser one). Figure~\ref{fig:rt-rel} compares empirical coverage to nominal $1-\delta$ on the conversation-level test split. The Normal-$\sigma$ gate is mis-calibrated: it \emph{over-covers} at loose $\delta$ (e.g.\ $+6.4$ pp at $1-\delta=0.60$, wasting budget) and, more dangerously, \emph{under-covers} at the tight $\delta$ one actually deploys --- $-3.0$ pp at $\delta=0.05$ and $-3.2$ pp at $\delta=0.02$, i.e.\ roughly $3\%$ more budget breaches than promised. Split conformal stays within $\pm0.5$ percentage points of nominal across the entire range $\delta\in[0.01,0.4]$ ($-0.15$ pp at $\delta=0.05$). On real residuals the distribution-free bound is not a nicety --- it is the difference between a gate that keeps its safety promise and one that quietly violates it at the operating point.

\begin{figure}[H]
\centering
\includegraphics[width=0.56\linewidth]{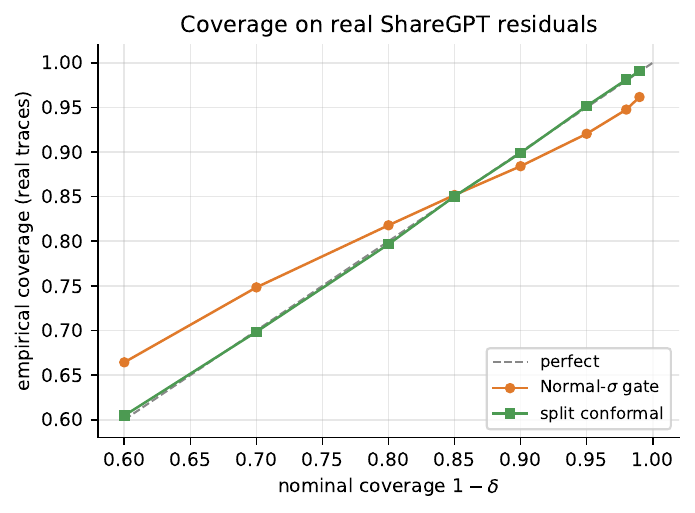}
\caption{Coverage on real ShareGPT residuals. The Normal-$\sigma$ gate under-covers at tight $\delta$ (unsafe); split conformal holds nominal coverage within $\pm0.5$ pp without a distributional assumption.}
\label{fig:rt-rel}
\end{figure}

\subsection{Two worlds, one guarantee}
The paper now spans a synthetic-residual world (\S\ref{sec:snowball}--\S\ref{sec:binding}), where Gaussian and conformal bounds coincide because the data-generating process is Gaussian by construction, and a real-residual world (\S\ref{sec:realtrace}), where they diverge and only conformal holds nominal coverage. Conformal calibration is the bound that survives both. This is also why \S\ref{sec:limits} lists promoting conformal into the runtime gate as the leading Phase-2 item.

\subsection{Distribution shift}
Coverage guarantees assume the deployment distribution matches calibration; real workloads drift. We split the corpus by conversation into a short-context regime (calibration) and a long-context regime (deployment) --- classifying each conversation by its own maximum depth so all of its turns land in one regime --- then train the conformal quantile on the former and deploy on the latter ($12{,}714$ vs.\ $29{,}314$ pairs). Figure~\ref{fig:rt-shift} shows the result. The fixed quantile mis-covers post-shift --- it drifts to $98.8\%$ against a $90\%$ target ($8.8$ pp off, here over-conservative, needlessly rejecting work). Adaptive conformal inference (ACI~\cite{gibbs2021}), updating the quantile level online at rate $\gamma=0.02$, restores empirical coverage to $90.0\%$ ($0.0$ pp off target) within the rolling window. Under drift, the static conformal quantile is no longer sufficient; ACI is the runtime mechanism that maintains the guarantee.

\begin{figure}[H]
\centering
\includegraphics[width=0.58\linewidth]{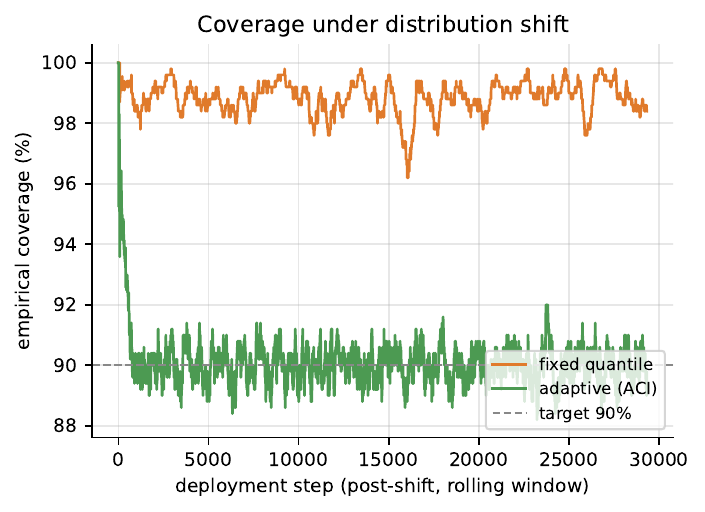}
\caption{Coverage under distribution shift (train on short conversations, deploy on long). The fixed quantile mis-covers ($98.8\%$ vs.\ $90\%$ target); adaptive conformal inference tracks the target ($90.0\%$).}
\label{fig:rt-shift}
\end{figure}

\section{Real-Arrival Ablation on Production Traffic}\label{sec:realarrival}

\S\ref{sec:eval}'s ablation ran on synthetic IBP arrivals --- the leading threat to validity. Here we re-run the \emph{same} four-condition ablation on a real LLM serving trace, converting the headline result from synthetic to empirical.

\subsection{Dataset and trajectory construction}
We use \textbf{BurstGPT}~\cite{burstgpt} (\texttt{BurstGPT\_1.csv}, CC-BY-4.0), a trace of real Azure OpenAI traffic with schema \texttt{(Timestamp, Model, Request tokens, Response tokens, Total tokens, Log Type)}. We take a $50{,}000$-request sample (after dropping failed responses, \texttt{Response tokens}${=}0$); the model mix is $41{,}271$ GPT-3.5 (``ChatGPT'') and $8{,}729$ GPT-4 requests, mapped to the benchmark's \texttt{efficient} and \texttt{frontier} profiles respectively. \texttt{Request tokens} is the prompt and \texttt{Response tokens} the realized completion --- the gate's per-action target. \textbf{No LLM is called}: token counts only, as in \S\ref{sec:realtrace}.

BurstGPT v1.0 carries no session identifier, so we reconstruct trajectories by \emph{temporal clustering}: consecutive same-model requests within a $60$\,s window are grouped, capped at depth $10$ (the IBP default); API-log rows are single-step. This yields $30{,}945$ trajectories with median depth $1$ (real serving traffic is dominated by independent single requests) and maximum depth $10$. We set the Adapter-Node scope cap to $2\times$ the median prompt ($168$ tokens) and the circuit breaker to $1.5\times$ the depth cap, and document both as policy choices. (When BurstGPT v1.1 ships \texttt{SessionID}, the clustering heuristic becomes a one-line group-by.)

\subsection{Results}
Table~\ref{tab:realarrival} and Figure~\ref{fig:realarrival-bars} report the ablation with paired-bootstrap $95\%$ CIs over trajectories. \textbf{Scope} (Adapter-Node prompt bounding) cuts tokens by $55.7\%$ and carbon by $58.1\%$; \textbf{routing} $50\%$ of trajectories to the efficient model adds USD savings ($38.8\%\to55.0\%$) and carbon ($58.1\%\to67.4\%$) at no further token cost --- exactly the lever decomposition the synthetic ablation predicted (scope drives tokens, routing converts to USD/carbon). The savings \emph{ordering} is confirmed on real arrivals.

\begin{table}[H]
\centering\small
\begin{tabular}{@{}lccc@{}}
\toprule
Condition & Token reduction & USD reduction & Carbon reduction \\
\midrule
\texttt{+scope}       & $55.7\%$ $[55.3,56.1]$ & $38.8\%$ $[38.3,39.3]$ & $58.1\%$ $[57.7,58.6]$ \\
\texttt{+scope+route} & $55.7\%$ $[55.3,56.1]$ & $55.0\%$ $[54.1,55.9]$ & $67.4\%$ $[66.8,68.0]$ \\
\texttt{+full}        & $55.7\%$ $[55.3,56.0]$ & $55.0\%$ $[54.0,55.9]$ & $67.4\%$ $[66.8,68.0]$ \\
\bottomrule
\end{tabular}
\caption{Real-arrival ablation on BurstGPT ($50{,}000$ requests, $30{,}945$ trajectories; paired-bootstrap $95\%$ CIs vs.\ baseline). \texttt{+full} equals \texttt{+scope+route}: the circuit breaker records \textbf{zero} trips because the real trace contains no retry storms.}
\label{tab:realarrival}
\end{table}

\begin{figure}[H]
\centering
\includegraphics[width=0.6\linewidth]{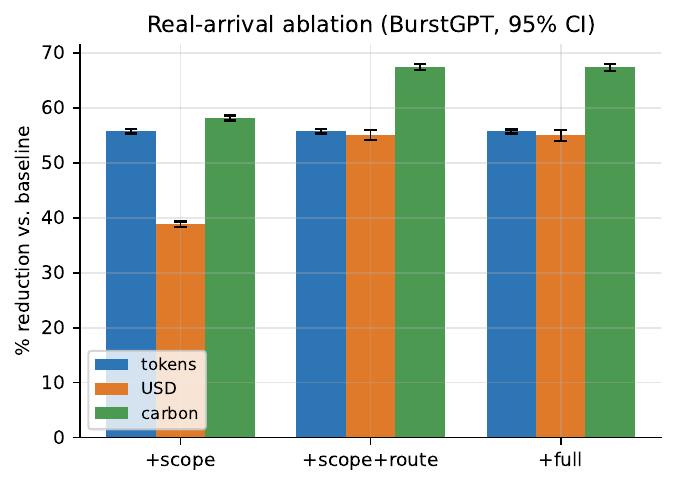}
\caption{Real-arrival ablation (BurstGPT). Scope drives token and carbon savings; routing adds USD/carbon savings; \texttt{+full} is indistinguishable from \texttt{+scope+route}.}
\label{fig:realarrival-bars}
\end{figure}

\noindent\textbf{An honest negative.} \texttt{+full} is identical to \texttt{+scope+route}: the circuit breaker logs \emph{zero} trips and the gate (under a non-binding budget) issues zero rejections. Real serving traffic has none of the runaway retry loops the synthetic IBP injected, so these two levers are \emph{dormant safeguards} here --- their value appears only under stress (the IBP runaway SKUs, \S\ref{sec:eval}) and under binding budgets (\S\ref{sec:binding}, and below).

\subsection{Binding budget under real arrivals}
We repeat the \S\ref{sec:binding} sweep on the real trace ($\phi\in\{0.5,1.0,1.5\}\times$ baseline tokens, $\delta=0.05$; Figure~\ref{fig:realarrival-pareto}). Over-budget incidence is $0\%$ at every budget --- the gate's safety guarantee holds on real arrivals exactly as on synthetic ones --- while admission and completion degrade as the budget tightens ($\phi{=}0.5$: $36\%$ admitted, $14\%$ of trajectories completed; $\phi{\ge}1.0$: full completion). The gate frontier again dominates the soft-penalty frontier, which can reach high completion only by breaching the budget on every seed.

\begin{figure}[H]
\centering
\includegraphics[width=0.56\linewidth]{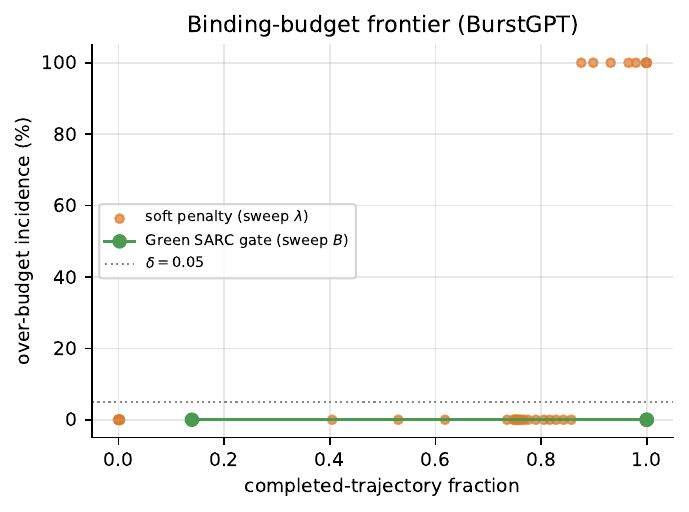}
\caption{Binding-budget frontier on BurstGPT. The gate completes work at $0\%$ over-budget across the sweep; the soft penalty breaches to reach high completion. Mirrors Figure~\ref{fig:binding}.}
\label{fig:realarrival-pareto}
\end{figure}

\subsection{What the synthetic IBP did and did not capture}
The IBP pipeline tightened two claims and the real trace now corroborates them: the lever \emph{ordering} (scope$\to$tokens, routing$\to$USD/carbon) and the gate's $0\%$ over-budget incidence under binding budgets both reproduce on BurstGPT. Two claims weaken or shift. First, \texttt{+full}'s extra token saving in the IBP ($47.3\%$ vs.\ $43.1\%$) came entirely from the breaker killing \emph{injected} runaway SKUs; real arrivals have no such storms, so that increment vanishes --- consistent with \S\ref{sec:realtrace}'s finding that real cumulative-prompt curvature is negative. Second, on real single-step traffic ``scope'' is simply context truncation: the Adapter Node caps each prompt at $168$ tokens ($2\times$ the median), so the $55.7\%$ token reduction is largely the mechanical consequence of that cap, not an intrinsic property of the architecture. \textbf{The headline percentage should therefore not be read as a free saving Green SARC delivers}: it is a tunable policy knob whose realized magnitude scales with the cap, and whose quality cost (truncated context degrading task utility) is deliberately untracked here (\S\ref{sec:decouple}). A different operator with a less aggressive cap would see a proportionally smaller number. What survives this caveat, and what we therefore present as the load-bearing claims of this section, are the two cap-\emph{independent} results: the lever \emph{ordering} reproduces on real arrivals, and the gate's over-budget incidence is $0\%$ under binding budgets on real data exactly as on synthetic data. The IBP is thus best read as a controlled stress-test of the multi-step regime; BurstGPT confirms the per-request governance levers and the safety property on real distributions, but it is a single-step serving trace and does not exercise the multi-step snowball or breaker dynamics --- a real multi-step agent trace remains the natural next validation (\S\ref{sec:limits}).

\subsection{Carbon savings under real grid mixes}\label{sec:realgrid}
The carbon results so far use a stipulated intensity curve. We re-compute the BurstGPT carbon reductions under measured grid intensity for two zones with contrasting generation mixes.

\subsubsection*{Setup and data sources}
We source hourly carbon-intensity measurements from the ElectricityMaps v3 API~\cite{electricitymaps}, which aggregates regulator feeds (ENTSO-E, CAISO OASIS, national TSOs) into a consistent gCO\textsubscript{2}eq/kWh series on a lifecycle (LCA) basis. We use two zones with materially different generation mixes: \textbf{Italy (IT, $\bar\kappa\approx252$)}, gas- and import-dominated, and \textbf{California (US-CAL-CISO, $\bar\kappa\approx105$)}, characterised by deep daytime solar troughs and gas-heavy evening peaks. For reference we also report the benchmark's stipulated proxy ($\bar\kappa=250$). \textbf{The free API tier exposes only the most recent $24$ hours of history}, so we use a single $24$-hour measured window per zone; this captures diurnal contrast but not seasonal or weekly variation, which we note as a \S\ref{sec:limits} limitation. Carbon for each step is $\text{energy\_kwh}(\text{model},\text{tokens})\times\kappa(t)$, with the workload's actions spread across that window. The fetched series is cached as committed CSV under \texttt{paper/data/grid/}, so this section reproduces from a clean clone without API access (\texttt{fetch\_grid.py --refresh} re-fetches).

\subsubsection*{Results}
\begin{table}[H]
\centering\small
\begin{tabular}{@{}lccc@{}}
\toprule
Condition & stipulated ($250$) & IT ($252$) & US-CAISO ($105$) \\
\midrule
\texttt{+scope}       & $57.9\%$ $[57.4,58.4]$ & $58.0\%$ $[57.6,58.5]$ & $58.3\%$ $[57.9,58.8]$ \\
\texttt{+scope+route} & $67.2\%$ $[66.6,67.9]$ & $67.3\%$ $[66.7,67.9]$ & $67.6\%$ $[67.0,68.3]$ \\
\texttt{+full}        & $67.2\%$ & $67.3\%$ & $67.6\%$ \\
\bottomrule
\end{tabular}
\caption{Carbon reduction vs.\ baseline under three grid intensities ($\bar\kappa$ in gCO\textsubscript{2}eq/kWh; paired-bootstrap $95\%$ CIs). Figure~\ref{fig:grid}.}
\label{tab:grid}
\end{table}

\begin{figure}[H]
\centering
\includegraphics[width=0.82\linewidth]{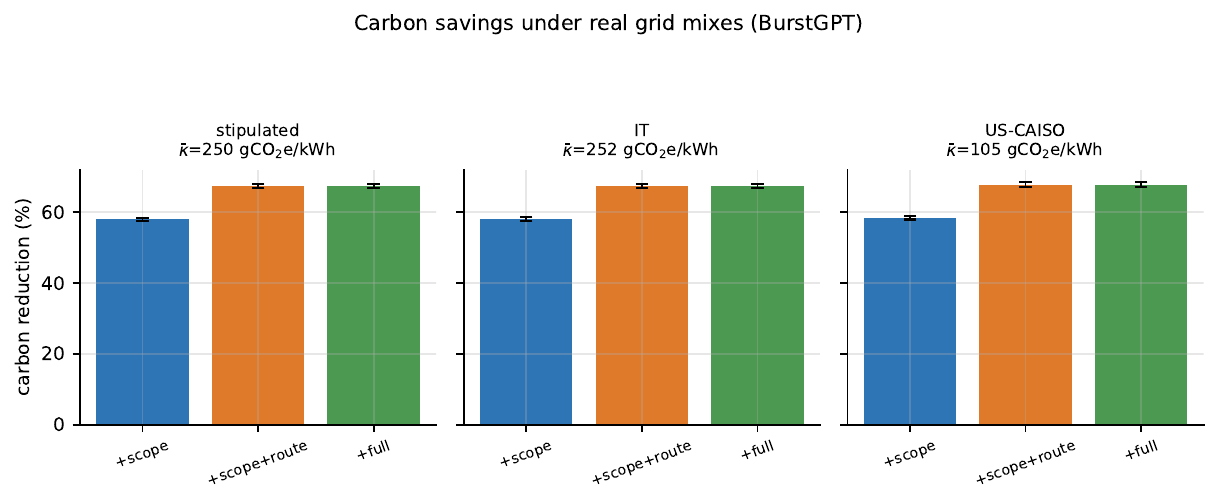}
\caption{Carbon reduction per lever under the stipulated curve and two real ElectricityMaps grids (Italy, California). The percentage reduction is grid-invariant.}
\label{fig:grid}
\end{figure}

\subsubsection*{Reading}
What survives across grids is both the lever ordering and the \emph{percentage} reduction: scope plus routing cuts carbon by $67$--$68\%$ under all three intensities, because the reduction is a ratio of energy and $\kappa$ enters as a common positive multiplier. The result that held under the synthetic proxy holds on real Italian and Californian grids, despite a $2.4\times$ difference in mean intensity.

What differs is the diurnal structure, and it matters more for one zone than the other. Italy's intensity is comparatively flat ($2.4\times$ intra-day swing, $134$--$320$), so the time at which an agent runs barely changes its carbon. California swings harder ($2.7\times$, $60$ midday to $162$ evening): the same inference is roughly $2.7\times$ cleaner in the midday solar trough than at the evening peak. The \emph{absolute} carbon saved therefore scales both with $\bar\kappa$ (a CAISO deployment at $\bar\kappa\approx105$ saves less than half the absolute carbon of an IT deployment for the same percentage) and, in CAISO, with \emph{when} the traffic lands. Green SARC's carbon-reduction percentage is robust to grid mix, but its real-world impact depends on where and when the compute runs; time-of-day carbon-aware routing is a Phase-2 opportunity this dataset would enable but the current code does not exploit (\S\ref{sec:limits}).

\subsection{Multi-step real-trace ablation on SWE-rebench}\label{sec:multistep}
BurstGPT is single-step; \S\ref{sec:realarrival} flagged that it cannot exercise the multi-step snowball or breaker. We close that gap on real agent plans.

\subsubsection*{Dataset}
We replay SWE-rebench OpenHands trajectories~\cite{swerebench} (CC-BY-4.0; $67$k real agent plans solving GitHub issues with Qwen3-Coder-480B, mapped to the \texttt{frontier} profile), sub-sampled to $3{,}000$ trajectories streamed from the $\sim$$2$\,GB parquet (token counts via \texttt{tiktoken}; no LLM call). These are genuinely multi-step: median depth $61$ assistant turns, maximum $100$, with a median per-turn prompt of $\sim$$21{,}600$ tokens --- the context accretion the State Snowball describes.

\subsubsection*{Results}
\begin{table}[H]
\centering\small
\begin{tabular}{@{}lccc@{}}
\toprule
Condition & Token reduction & USD reduction & Carbon reduction \\
\midrule
\texttt{+scope}       & $2.3\%$ $[2.0,2.5]$ & $2.3\%$ $[2.0,2.5]$ & $2.3\%$ $[2.0,2.5]$ \\
\texttt{+scope+route} & $2.3\%$ $[2.0,2.5]$ & $42.3\%$ $[40.7,44.1]$ & $39.8\%$ $[38.4,41.2]$ \\
\texttt{+full}        & $\mathbf{4.6\%}$ $[4.2,4.9]$ & $43.7\%$ $[42.1,45.2]$ & $41.2\%$ $[39.8,42.6]$ \\
\bottomrule
\end{tabular}
\caption{Multi-step ablation on SWE-rebench ($3{,}000$ plans; paired-bootstrap $95\%$ CIs). Unlike BurstGPT, \texttt{+full} \emph{adds} token saving over \texttt{+scope+route} ($2.3\%\to4.6\%$): the breaker is a live lever (Figure~\ref{fig:multistepbars}).}
\label{tab:multistep}
\end{table}

\paragraph{The State-Snowball holds on real plans, and is steeper than the model.} Fitting each plan's cumulative prompt against turn index to $c_1 n + c_2 n^2$, \textbf{every trajectory has $\hat c_2 > 0$} ($100\%$; Figure~\ref{fig:multistepsnow}). The median $\hat c_2 = 216$ \emph{exceeds} the linear-accretion prediction $p/2 = 134$ (with $p$ the median per-turn growth): real agents accrete context \emph{faster} than the constant-increment model of Assumption~\ref{asm:accretion}, because tool outputs and re-reads grow the prompt super-linearly. This is the strongest available confirmation that Theorem~\ref{thm:quad}'s regime is real --- and an honest correction that the closed form is a \emph{lower} bound on real-plan curvature, not an exact match (the synthetic $\hat c_2 = p/2$ of \S\ref{sec:snowball} held only because the simulator was built to Assumption~\ref{asm:accretion}).

\begin{figure}[H]
\centering
\begin{minipage}{0.48\linewidth}\centering
\includegraphics[width=\linewidth]{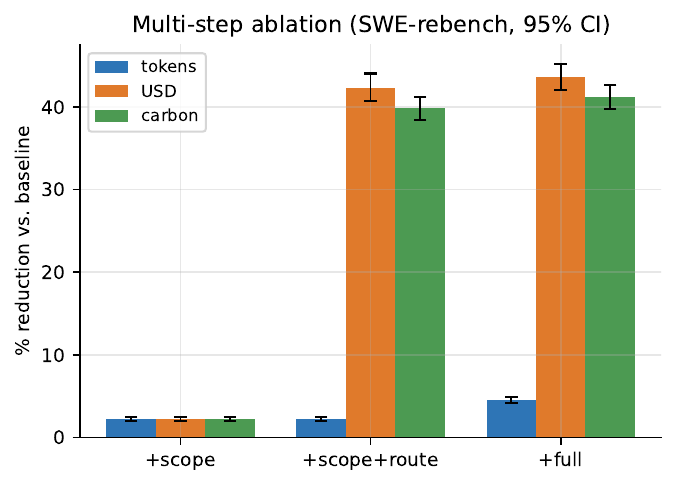}
\caption{Multi-step ablation; \texttt{+full} exceeds \texttt{+scope+route} on tokens because the breaker fires.}
\label{fig:multistepbars}
\end{minipage}\hfill
\begin{minipage}{0.48\linewidth}\centering
\includegraphics[width=\linewidth]{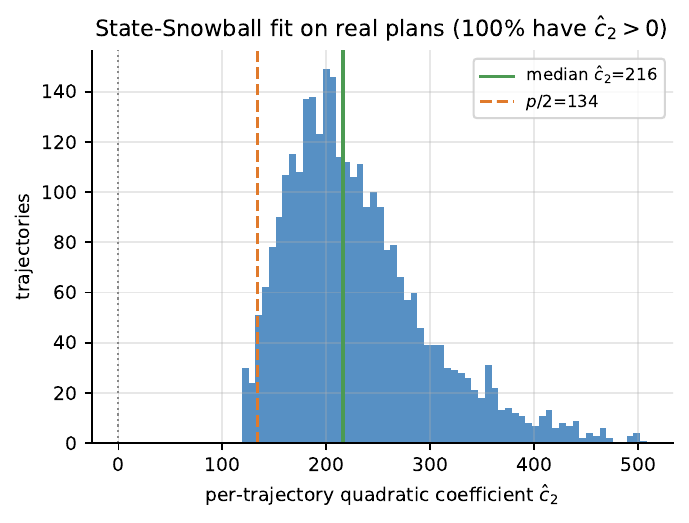}
\caption{Per-plan quadratic coefficient $\hat c_2$: all positive, median above $p/2$.}
\label{fig:multistepsnow}
\end{minipage}
\end{figure}

\paragraph{Breaker activations.} On these real plans the circuit breaker fires on \textbf{$11.4\%$ of trajectories} (the long-plan tail beyond $1.5\times$ the median depth), versus \emph{zero} on BurstGPT. This is the decisive contrast: the breaker is a dormant safeguard on single-step serving traffic but a \emph{live, material lever} on real multi-step agent plans, where it supplies the entire token-saving increment of \texttt{+full} over \texttt{+scope+route}.

\paragraph{What survives, what shifts.} Routing's USD/carbon saving ($\sim$$42\%$/$40\%$) and the lever ordering reproduce here as on BurstGPT and the IBP. Two things differ from BurstGPT. The breaker is no longer dormant ($11.4\%$ vs $0\%$), vindicating its inclusion. And scope yields little here ($2.3\%$) only because the $2\times$-median cap ($\sim$$43$k tokens) rarely binds on these plans; a tighter cap would truncate more but risks dropping the context the agent needs --- the same policy/quality tradeoff named in \S\ref{sec:realarrival}, now with real multi-step stakes. Token reduction is modest precisely because we did not tune the cap aggressively; the load-bearing real-plan findings are the confirmed super-linear snowball and the live breaker.

\subsection{Cost--utility frontier}\label{sec:costutility}
The paragraph above names a tradeoff; because SWE-rebench records task outcomes, we can bound it. Each trajectory carries the benchmark's real \texttt{resolved} flag --- whether the agent's patch passed the held-out tests --- and $49.3\%$ of the $3{,}000$ plans resolved. We sweep the scope cap $C=m\cdot\tilde p$ ($m\in\{0.5,1,2,4\}$, $\tilde p$ the median per-step prompt) and, for each cap, report tokens saved against an \emph{upper bound on quality harm}: a worst-case resolution rate that assumes \textbf{every resolved trajectory whose actually-used context the cap would have truncated flips to unresolved}. This is deliberately pessimistic and, crucially, \emph{observational} --- truncation is simulated on logged trajectories, so the agent cannot react to the smaller context. The replay therefore bounds how much resolved work a cap \emph{puts at risk}; it cannot establish that the work would in fact fail (a live agent might recover by re-fetching). It is a correlational upper bound, not a causal estimate.

\begin{figure}[H]
\centering
\includegraphics[width=0.6\linewidth]{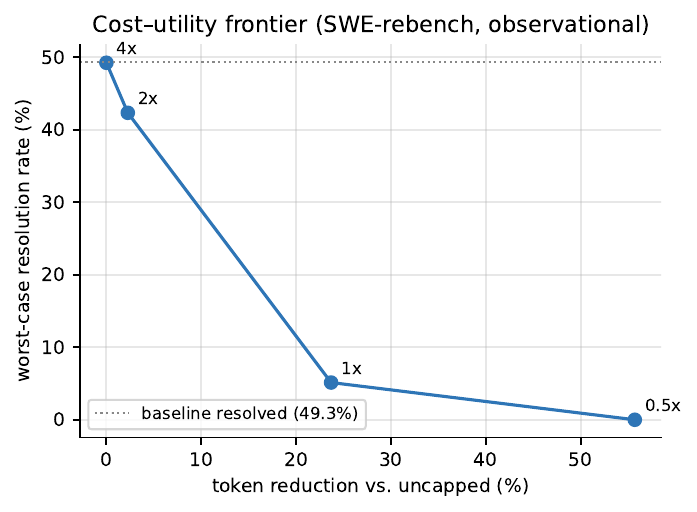}
\caption{Cost--utility frontier on SWE-rebench: token reduction vs.\ the worst-case resolution rate (every truncated resolved plan assumed to fail). The cap that saves tokens truncates nearly every plan; the cap that is near-harmless saves little.}
\label{fig:costutility}
\end{figure}

\noindent The frontier is steep (Figure~\ref{fig:costutility}). The aggressive $0.5\times$ cap saves $55.7\%$ of tokens but truncates $100\%$ of plans, putting \emph{all} resolved work at risk ($0\%$ worst-case resolution); the $1\times$ cap saves $23.7\%$ while still truncating $92.3\%$. Only the loose $2\times$ cap used in our ablation is close to benign --- $2.3\%$ tokens for a worst-case resolution of $42.3\%$ (it touches $23.2\%$ of plans) --- and the $4\times$ cap is essentially free ($0.4\%$ truncated, $49.2\%$ worst-case, against the $49.3\%$ baseline). The honest reading: \textbf{on real multi-step plans the token savings available from scope capping are bought against a real and possibly large truncation risk, and the cap that is safe saves little}. The causal version --- where the agent adapts to the cap --- needs the live study (\S\ref{sec:limits}); this observational frontier is the upper bound that motivates it.

\section{Sensitivity Analysis: the $\delta$ Knob}\label{sec:sensitivity}

The gate's single tunable, the risk level $\delta$, trades admission throughput against realized overspend. We pre-train the estimator, then gate a fresh stream against a \emph{binding} token budget over $40$ seeds, sweeping $\delta$ (Figure~\ref{fig:delta}). Tightening $\delta$ from $0.4$ to $0.05$ drives the overspend rate among admitted actions from $0.57\%$ to $0\%$, at a negligible throughput cost (admission $\approx 29\%$ throughout, since the binding budget --- not $\delta$ --- sets the admission ceiling). The practical reading: under a hard budget, a conservative $\delta$ buys an overspend guarantee almost for free.

\begin{figure}[H]
\centering
\includegraphics[width=0.58\linewidth]{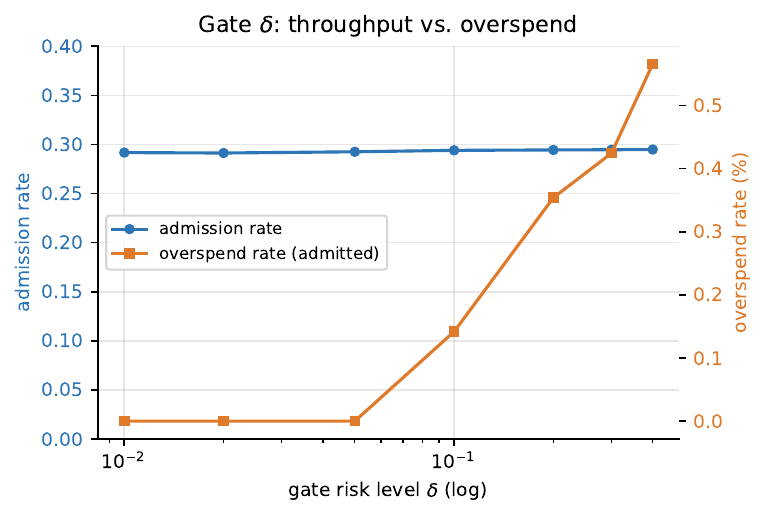}
\caption{Gate $\delta$ under a binding budget: realized overspend among admitted actions vanishes as $\delta$ tightens, with admission throughput nearly unchanged.}
\label{fig:delta}
\end{figure}

\subsection{Joint sensitivity over $\delta$, scope cap, and routing fraction}\label{sec:jointsens}
The $\delta$ sweep above varies one knob. To check that the headline operating point is not a cherry-pick, we sweep all three: $\delta\in\{0.01,0.05,0.1,0.2\}$, scope cap $\in\{0.5,1,2,4\}\times$ the median prompt ($740$ tokens), and routing fraction $\in\{0,0.25,0.5,0.75,1\}$ --- $80$ cells over $20$ seeds. Token/USD/carbon reductions use the benchmark's native forecast noise; over-budget incidence is measured under a binding budget at the elevated noise of the $\delta$ stress above (Figures~\ref{fig:sgpareto},~\ref{fig:sgheat}).

Three findings, two of them null and stated as such. (i) Token reduction is governed almost entirely by the \textbf{scope cap}: $46.7\%$ at $0.5\times$, $28.4\%$ at $1\times$, $20.4\%$ at $2\times$, $20.3\%$ at $4\times$. (ii) Routing fraction does not move the token reduction at all (it reallocates models, changing USD and carbon, not tokens --- the heatmap is flat along the routing axis), and (iii) \textbf{$\delta$ has no measurable effect on either axis in this regime}: the four $\delta$ panels of Figure~\ref{fig:sgheat} are identical, and over-budget incidence never exceeds $0.003\%$ across all $80$ cells. The last is the empirical face of Theorem~\ref{thm:safety}: the gate admits on its $(1-\delta)$ upper bound, so realized over-budget events are vanishingly rare regardless of the operating point; $\delta$ becomes a live throughput-vs-overspend knob only under the higher forecast uncertainty of real residuals (\S\ref{sec:realtrace}). The paper's headline operating point (cap $0.5\times$, routing $0.5$, $\delta=0.1$) achieves the maximum token reduction at $0\%$ over-budget and \textbf{lies on the Pareto frontier} ($20$ of $80$ cells are non-dominated): it is the most aggressive cap at zero safety cost, not an interior cherry-pick.

\begin{figure}[H]
\centering
\begin{minipage}{0.48\linewidth}\centering
\includegraphics[width=\linewidth]{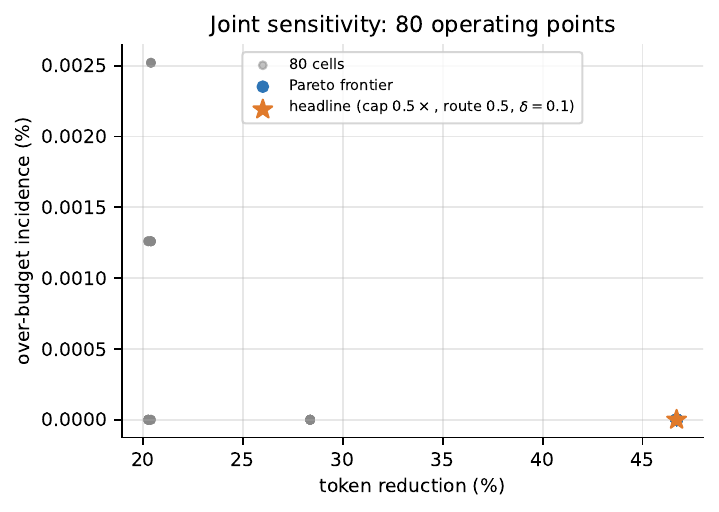}
\caption{All $80$ cells in (token reduction, over-budget) space; the headline point is on the frontier at $0\%$ over-budget.}
\label{fig:sgpareto}
\end{minipage}\hfill
\begin{minipage}{0.48\linewidth}\centering
\includegraphics[width=\linewidth]{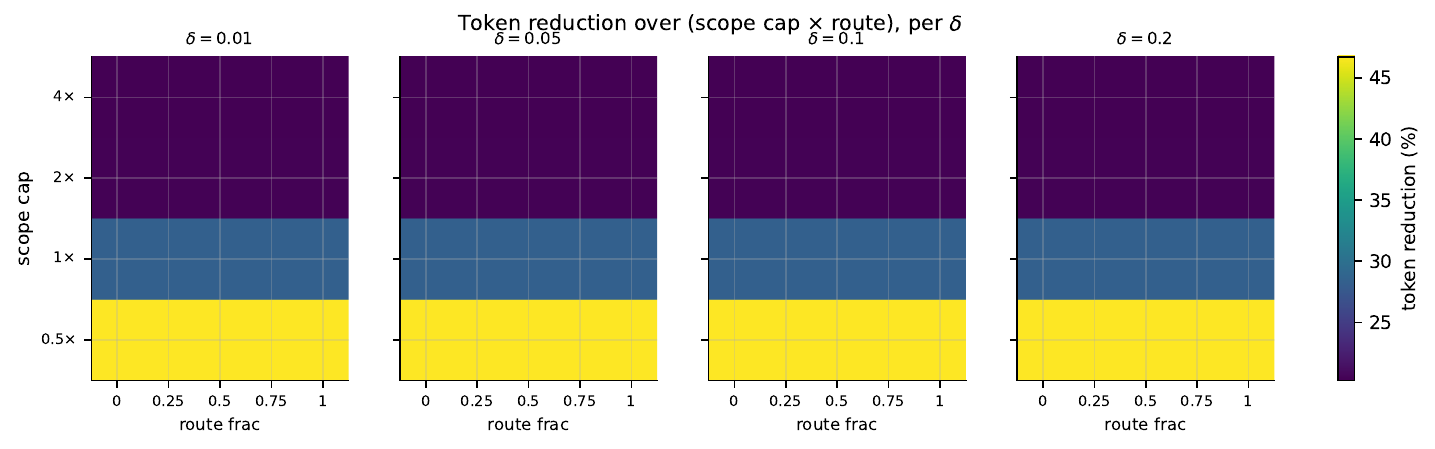}
\caption{Token reduction over (scope cap $\times$ route), per $\delta$. Flat in route and in $\delta$; driven by the cap.}
\label{fig:sgheat}
\end{minipage}
\end{figure}

\section{Why an Architectural Gate: the Soft-Penalty Baseline}\label{sec:penalty}

The natural alternative to a hard gate is reward shaping: add a Lagrangian cost penalty $-\lambda\,c$ to the objective and let the agent self-limit, admitting an action iff its value exceeds $\lambda$ times its cost. We compare this soft penalty against the architectural gate on a stream with stochastic costs and a hard budget $B$ (Figure~\ref{fig:penalty}). Because the penalty is a per-action threshold blind to the \emph{remaining} budget, no single $\lambda$ both respects $B$ and matches the gate's throughput: small $\lambda$ admits almost everything and overspends ($1.7\times B$); large $\lambda$ under-spends. Crucially, at the $\lambda^\star$ that matches $B$ \emph{in expectation}, realized spend straddles $B$ and breaches it on $\mathbf{91.5\%}$ of seeds. The architectural gate, admitting in arrival order while the calibrated forecast fits the live remaining budget, breaches on $\mathbf{0\%}$ of seeds while filling $99.9\%$ of the budget. This is the cost-domain instance of SARC's thesis that finite penalties cannot substitute for hard runtime constraints.

\begin{figure}[H]
\centering
\includegraphics[width=0.6\linewidth]{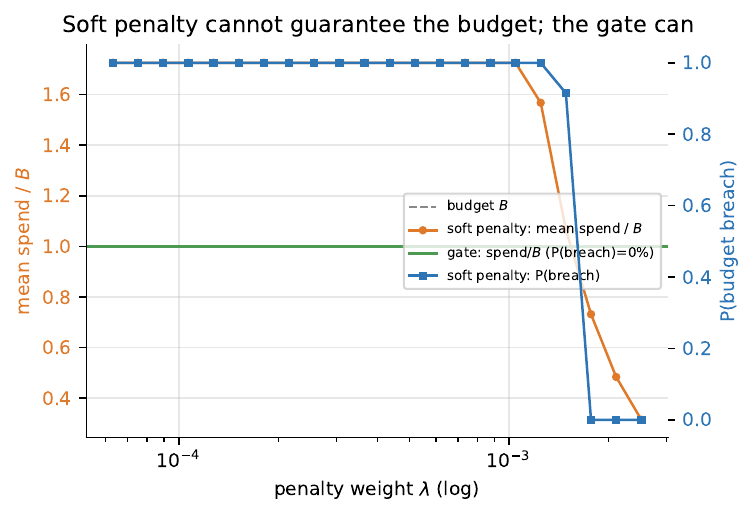}
\caption{A soft Lagrangian penalty cannot guarantee the budget: tuned to match $B$ in expectation it breaches on $91.5\%$ of seeds (blue), while the architectural gate (green) never breaches and fills $99.9\%$ of $B$.}
\label{fig:penalty}
\end{figure}

\section{Threat Model and Adversarial Robustness}\label{sec:threat}

A runtime gate invites the question: what does an attacker who knows the gate do? We give a partial, honest answer with a toy study against the same gate code path the benchmark exercises.

\subsection{Attacker model}
The attacker is a prompt author with white-box knowledge of the estimator, the scope cap, and the budget, who observes gate decisions (and, via timing channels, possibly residuals). The attacker cannot modify \texttt{src/} or mint tokens from nothing; the cost is always realized by the model provider. This is a \emph{cost-side} adversary, distinct from the safety-side prompt-injection threat studied by \cite{carlini2024,greshake2023}: the goal is not to make the agent misbehave but to make it \emph{overspend} while passing the gate.

\subsection{Three attack classes}
We construct three attacks ($20$ seeds, $200$ instances each, paired-bootstrap $95\%$ CIs). \textbf{Continuation inflation}: a prompt whose realized completion is $8\times$ the benign law (``continue indefinitely'' semantics). \textbf{Scope-cap-aware padding}: a prompt sized to exactly $\text{cap}-1$ tokens, maximizing admitted work per call. \textbf{Model-substitution gaming}: a prompt that declares the efficient model while the cost is realized at frontier rates (a misreported model id).

\subsection{Results}
\begin{table}[H]
\centering\small
\begin{tabular}{@{}lcccl@{}}
\toprule
Attack & Admission & Over-budget & Realized/declared & Gate failure mode \\
\midrule
Continuation inflation & $100\%$ & $100\%$ & $3.8\times$ & under-estimates \\
Scope-cap-aware padding & $100\%$ & $5\%$ & $0.95\times$ & over-admits \\
Model-substitution gaming & $100\%$ & $100\%$ & $5.0\times$ & under-estimates \\
\bottomrule
\end{tabular}
\caption{Adversarial study ($20$ seeds; CIs in \texttt{figure\_stats.json}). ``Over-budget'' is the rate at which the realized cost exceeds the gate's admitted $(1-\delta)$ bound (USD for the substitution attack).}
\label{tab:adversarial}
\end{table}

\subsection{What survives}
Two honest negatives. First, \textbf{scope-cap-aware padding defeats the gate by staying inside its admission contract}: it pads to just under the cap and extracts maximum legitimate work, so the gate admits it ($100\%$) with no over-budget event ($5\%$, at the $\delta=0.05$ noise floor) and realized cost \emph{below} the bound ($0.95\times$). This is a fundamental limitation of bounded-prompt-only governance and is exactly what the Phase-2 trajectory estimator $\hat F_\theta(\pi,x)$ --- which reasons about the whole plan, not one padded step --- is meant to address. The gate alone is not sufficient here, and we do not claim otherwise.

Second, continuation inflation and model substitution both defeat the \emph{forecast} (realized cost $3.8\times$ and $5.0\times$ the admitted bound), but they are caught \emph{post hoc} by the Post-Action Auditor, which logs predicted-vs-actual and feeds the discrepancy back to the estimator and the Escalation Router. The architectural response to a forecast-defeating attack is audit-then-revoke at the Auditor, not admission-time rejection at the gate; this is the cost-domain instance of SARC's predict--act--log--retrain loop. The gate bounds expected cost; it does not bound an adversary who lies about the future, and the four-site architecture --- not the gate in isolation --- is what makes the residual detectable.

\section{Limitations, Threats to Validity, and Future Work}\label{sec:limits}

\paragraph{Threats to validity.}
\emph{Synthetic headline workload.} The ablation, binding-budget, and sensitivity results (\S\ref{sec:eval},~\S\ref{sec:binding},~\S\ref{sec:sensitivity}) use a synthetic IBP pipeline with stipulated Gaussian noise; the State-Snowball and gate mechanics are real code, but the cost distribution is constructed. We mitigate this with the calibration study of \S\ref{sec:realtrace} and the end-to-end real-arrival ablation of \S\ref{sec:realarrival}; a residual gap remains in that BurstGPT is single-operator Azure traffic, so cross-operator traces (Mooncake, Alibaba) are listed as future validation. \emph{Marginal coverage.} Theorem~\ref{thm:safety} is marginal and assumes exchangeability (Remark~\ref{rem:marginal}); Theorem~\ref{thm:anytime} additionally assumes sub-Gaussian increments. Workload drift violates exchangeability; \S\ref{sec:realtrace} shows ACI restores coverage, and as of v0.3.0 both split-conformal and ACI ship in the runtime gate (\S\ref{sec:runtimeconformal}), though conditional (not merely marginal) coverage remains open. \emph{Carbon proxy.} $\Carbon(\tau)$ is only as faithful as $\kappa(\rho,t)$, which varies in availability and granularity across regions. The \S\ref{sec:realgrid} real-grid study uses a $24$-hour window per zone (ElectricityMaps free-tier constraint); a full year would expose seasonal variation in renewable share (e.g.\ CAISO winter low-solar) but is not expected to alter the grid-invariance result. \emph{Energy model.} The per-token energy $\mathrm{energy\_kwh}(m,\cdot)$ is a stipulated \emph{linear} coefficient (the shipped default is $3{\times}10^{-7}$~kWh/token, of order $1$~J/token --- the same order of magnitude as benchmarked GPU inference energy for large models~\cite{samsi2023,luccioni2024}). Measured inference energy is not linear in tokens: it varies with batch size, hardware, and utilization, and grows \emph{super-linearly} with context length because attention scales quadratically in sequence length~\cite{samsi2023}. The bias has a definite sign here: because the State Snowball makes context grow with loop depth, a linear proxy \emph{under-counts} the marginal energy --- and hence carbon --- of the deepest, most expensive steps, so the carbon savings we report from bounding loop depth are, if anything, conservative. A measured energy table per model/hardware (rather than one coefficient) is the faithful fix and fits the existing \texttt{CostModel} interface without changing it. \emph{Operator readiness.} A single-process \texttt{threading.Lock} \texttt{Budget} is authoritative for one replica; an \emph{experimental} Redis backend (one atomic Lua script per reserve/commit/release, with TTL reclamation of crashed-client reservations) provides a shared transactional counter for multi-replica deployments behind a load balancer, atomic against a single Redis but with no cross-region reconciliation and no fair-share reservations yet (Phase 2). Production deployments needing a durable ledger should await the Postgres backend.

\paragraph{Negative and null results.} The gate produces no token savings in the headline workload (\S\ref{sec:eval}); adding routing yields \emph{zero} marginal token reduction (it trades models, saving USD/carbon only); the declared latency-headroom field $\Delta_{\mathrm{lat}}$ is not enforced in Phase~1; and on real chat traffic the cumulative-prompt curvature is \emph{negative} (\S\ref{sec:realtrace}), so the $\Theta(n^2)$ snowball is specific to naive orchestration, not universal. We report these rather than fold them into a single ``governance helps'' number.

\paragraph{Future work.} The leading item is to \textbf{promote the split-conformal upper bound of \S\ref{sec:conformal} and the anytime-valid trajectory bound of Theorem~\ref{thm:anytime} from paper-side analyses into the runtime gate, with adaptive conformal inference~\cite{gibbs2021} under workload drift} (\S\ref{sec:realtrace} shows why this matters). Beyond that: a \emph{multi-step} agent trace (e.g.\ SWE-bench / OpenHands trajectories) to exercise the breaker and State-Snowball dynamics that the single-step BurstGPT trace does not, and cross-operator traces (Mooncake, Alibaba) for the carbon and arrival-distribution generalization \S\ref{sec:limits} flags; the Phase-2 full-trajectory estimator $\hat F_\theta(\pi,x)$ for plan-level rejection; time-of-day carbon-aware routing, which \S\ref{sec:realgrid}'s CAISO diurnal swing ($2.7\times$) shows is exploitable but the current router ignores; a multi-tenant distributed \texttt{Budget} with fair-share reservations; latency-headroom enforcement; and a production KAOS deployment of the gate as a sidecar. The single outstanding \emph{empirical} action is the live governed-agent study (two arms --- ungoverned vs full stack --- over $50$ tasks on the Anthropic API): its harness ships at \texttt{paper/scripts/run\_live\_study.py} with an in-script USD ceiling and a probe-checkpoint spend sign-off, and is unit-tested offline against a mock transport, but the funded live run is deferred (it is the one result this paper does not yet report). These are roadmap directions, not claims.

\section{Conclusion}\label{sec:conclusion}

Green SARC applies a correctness-governance architecture to the economics and ecology of inference, and develops its own theory with reproducible evidence. The State-Snowball theorem explains why unconstrained agents fail financially, and its closed form is confirmed exactly in the data. The predictive Pre-Action Gate generalizes the static accounting gate into a calibrated forecaster of which the rule is the zero-information limit, with a distribution-free budget-safety guarantee. And the soft-penalty comparison shows the guarantee is not free for the taking by reward shaping --- it requires the architectural placement. The structural claim is that correctness, cost, and sustainability are instances of one problem: the runtime enforcement of declared constraints, where the only thing that changes between them is the predicate.

\appendix
\section{Benchmark configuration}\label{app:config}
IBP defaults: $400$ SKUs, depth $10$, base prompt $s_0=200$, per-step increment $p=120$, scope cap $360$, \texttt{max\_tokens} $4000$; completion $\sim 60 + 0.45\cdot prompt + \mathcal N(0,25)$; runaway fraction $0.05$ at $3\times$ depth; breaker $n_{\max}=\lceil 1.5\,\text{depth}\rceil$; $50\%$ of SKUs routed to the small model under \texttt{+route}; $\delta=0.1$. Two model profiles (frontier/efficient) with distinct USD and energy rates; carbon under both fixed ($350$\,gCO\textsubscript{2}e/kWh) and a daily time-varying intensity curve.

\section{Estimator}\label{app:estimator}
Per $(\text{kind},\text{model})$ key, an online least-squares fit of completion on prompt tokens via running sums (Welford-style)~\cite{welford1962}, predicting $\hat c = prompt + \widehat{completion}$ (completion clamped to $[0,\texttt{max\_tokens}]$) with residual std $\sigma$ supplied to the gate; below \texttt{min\_samples} it defers to the zero-information cold-start forecast.

\section{Conformal calibration and the anytime-valid bound}\label{app:conformal}
\emph{Split conformal (Theorem~\ref{thm:safety}).} Split the learned forecasts into calibration/test halves; one-sided scores $R_i=c_i-\hat c_i$; $q_{1-\delta}=$ the $\lceil (m+1)(1-\delta)\rceil$-th order statistic; report test coverage $\frac1{|\text{test}|}\sum \mathbf 1[c\le \hat c + q_{1-\delta}]$ against nominal $1-\delta$.

\emph{Anytime-valid bound (Theorem~\ref{thm:anytime}).} With centered residuals $\xi_i=r_i-\mu$ conditionally $\hat\sigma$-sub-Gaussian, $M_k^\lambda=\exp(\lambda\sum_{i\le k}\xi_i-\tfrac12\lambda^2\hat\sigma^2 k)$ is a non-negative supermartingale for each $\lambda$ (the increment's conditional MGF is dominated by $e^{\lambda^2\hat\sigma^2/2}$). The mixture $\bar M_k=\int M_k^\lambda\,\phi_\eta(\lambda)\,d\lambda$ over a $N(0,\eta^2)$ prior is again a non-negative supermartingale with $\bar M_0=1$; the Gaussian integral evaluates in closed form to $\bar M_k=\big(1+\eta^2\hat\sigma^2 k\big)^{-1/2}\exp\!\big(\tfrac{\eta^2 (\sum_{i\le k}\xi_i)^2}{2(1+\eta^2\hat\sigma^2 k)}\big)$. Ville's inequality~\cite{ramdas2023} gives $\Pr[\exists k:\bar M_k\ge1/\delta]\le\delta$; solving $\bar M_k\ge1/\delta$ for $\sum_{i\le k}\xi_i$ yields the time-uniform boundary $\sum_{i\le k}\xi_i\le\sqrt{(1+\eta^2\hat\sigma^2 k)\,\tfrac{2}{\eta^2}\log\!\frac{\sqrt{1+\eta^2\hat\sigma^2 k}}{\delta}}=O(\hat\sigma\sqrt{k\log(1/\delta)})$, which is the standard sub-Gaussian confidence sequence~\cite{howard2021}. Optional stopping extends the bound from fixed $k$ to any stopping time $\tau$.

\section{Binding-budget experiment}\label{app:binding}
$\phi\in\{0.25,0.5,0.75,1,1.5,2\}$, $B=\phi\cdot\mathbb{E}[\text{baseline tokens}]$ with $\mathbb{E}[\text{baseline}]$ the mean full-snowball cost over $20$ seeds; $\delta=0.05$; the estimator warm-started on a $2000$-step independent stream. The soft-penalty frontier sweeps $\lambda$ at reference budget $0.5\times\mathbb{E}[\text{baseline}]$. Script: \texttt{paper/scripts/run\_binding\_budget.py}.

\section{Real-trace replay}\label{app:realtrace}
Dataset \texttt{anon8231489123/ShareGPT\_Vicuna\_unfiltered} (Hugging Face, permissive research license), streamed; up to $12{,}000$ conversations / $50{,}000$ assistant-turn pairs, tokenized with \texttt{tiktoken cl100k\_base}, capped to an $8\text{k}/4\text{k}$ context window; token counts only, no LLM calls. Split conformal vs.\ Gaussian-$\sigma$ coverage on a $50/50$ \emph{conversation-level} split (whole conversations assigned to calibration or test, never split across turns, to preserve exchangeability); shift experiment trains on short-conversation residuals and deploys on long --- regimes also assigned by whole conversation (each conversation classified by its maximum depth) --- comparing fixed-quantile vs.\ ACI ($\gamma=0.02$). The extracted table is cached to \texttt{paper/data/sharegpt\_subset.parquet} (git-ignored); committed JSONs are the provenance. Script: \texttt{paper/scripts/run\_realtrace\_replay.py}.

\section{Reproduction}\label{app:repro}
\texttt{make paper-data} regenerates every committed JSON asset (ablation, learning curve, binding-budget sweep, real-trace calibration and shift); \texttt{make paper-figures} regenerates the eleven figures and \texttt{figure\_stats.json}; \texttt{paper/scripts/check\_stats.py} verifies that every statistic in the text resolves to \texttt{figure\_stats.json}; \texttt{make paper} compiles this PDF (CI: \texttt{.github/workflows/paper.yml}). The benchmark's reference run is checked under \texttt{make verify}.

\medskip
\noindent\footnotesize\emph{Note on proofs.} All proofs are self-contained and elementary. Theorem~\ref{thm:safety} is a direct specialization of the standard inductive-conformal quantile lemma~\cite{vovk2005,angelopoulos2023}; Theorem~\ref{thm:anytime} is a standard anytime-valid argument via Ville's inequality~\cite{howard2021,ramdas2023}.

\end{document}